%% file: 0main.tex
\documentclass[sn-mathphys,Numbered,final]{sn-jnl}

\usepackage{amsmath,amssymb,amsfonts}%
\usepackage{amsthm}%
\usepackage{mathrsfs}%
\usepackage{mathtools}
\usepackage{multirow}%
\usepackage[title]{appendix}%
\usepackage{textcomp}%
\usepackage{manyfoot}%
\usepackage{booktabs}%
\usepackage{algorithm}%
\usepackage{algorithmic}

\usepackage{listings}%

\input{0preamble}

\newtheorem{theorem}{Theorem}
\newtheorem{lemma}[theorem]{Lemma}
\newtheorem{corollary}[theorem]{Corollary}
\newtheorem{example}[theorem]{Example}%
\newtheorem{remark}[theorem]{Remark}%

\newtheorem{definition}[theorem]{Definition}%

\begin{document}

\title[(Un)Solvable Loop Analysis]{(Un)Solvable Loop Analysis}


\author{\fnm{Daneshvar} \sur{Amrollahi}}\email{daneshvar.amrollahi@tuwien.ac.at}

\author{\fnm{Ezio} \sur{Bartocci}}\email{ezio.bartocci@tuwien.ac.at}

\author*{\fnm{George} \sur{Kenison}}\email{george.kenison@tuwien.ac.at}

\author{\fnm{Laura} \sur{Kovács}}\email{laura.kovacs@tuwien.ac.at}

\author{\fnm{Marcel} \sur{Moosbrugger}}\email{marcel.moosbrugger@tuwien.ac.at}

\author{\fnm{Miroslav} \sur{Stankovič}}\email{miroslav.stankovic@tuwien.ac.at}

\affil{\orgname{TU Wien}, \orgaddress{ \city{Vienna}, \country{Austria}}}


\abstract{
Automatically generating invariants, key to computer-aided analysis of probabilistic and deterministic programs and compiler optimisation, is a challenging open problem.
Whilst the problem is in general undecidable, the goal is settled for restricted classes of loops.
For the class of \emph{solvable} loops, introduced by Kapur and Rodr\'iguez-Carbonell in 2004, one can automatically compute invariants from closed-form solutions of recurrence equations that model the loop behaviour.
In this paper we establish a technique for invariant synthesis for loops that are not solvable, termed \emph{unsolvable} loops.
Our approach automatically partitions the program variables and identifies the so-called \emph{defective} variables that characterise unsolvability.
\hl{Herein we consider the following two applications.
First, we present  a novel technique that automatically synthesises  polynomials from {defective}  monomials, that admit closed-form solutions and thus lead to polynomial loop  invariants.
Second, given an unsolvable loop, we synthesise solvable loops with the following property: the invariant polynomials of the solvable loops are all invariants of the given unsolvable loop.  
} 
Our implementation and experiments demonstrate both the feasibility and applicability of our approach to both deterministic and probabilistic programs.


}

\keywords{Invariant Generation, Solvable Loop Synthesis, Algebraic Recurrences, Verification, Solvable Operators}

\maketitle



%
%
%
\newpage

\section*{Extension of Previous Work}

This paper is an extended version of the conference paper `Solving Invariant Generation for Unsolvable Loops' published at SAS 2022 \cite{amrollahi2022solving}.
We extended the text and results from the conference paper. 

In addition, this paper brings the following two new contributions when compared to our SAS 2022 paper  \cite{amrollahi2022solving}.
First, we introduce an  algorithmic approach (Algorithm~\ref{alg:loop-synthesis}) that synthesises solvable loops from unsolvable loops, such that the polynomial invariants of the solvable loop are also invariants of the unsolvable one. 
Second, we expand the list of unsolvable loops (with new Examples~\ref{ex:fib1}--\ref{ex:nagata}), drawn from the Mathematics and Physics literature;  our benchmarks yield a new test suite for evaluating loop analysis techniques and demonstrate the feasibility of our approach.

\input{01introduction}
\input{02preliminaries}

\input{03recurrences}
\input{04defective}
\input{05invariants}

\input{06loopsynthesis}

\input{07probabilistic}

\input{08applications}

\input{09experiments}
\input{10conclusion}

\section*{Declarations}


\bmhead{Acknowledgements}

This research was supported by the Vienna Science and Technology Fund WWTF 10.47379/ICT19018 grant ProbInG, the European Research Council Consolidator Grant ARTIST 101002685, the Austrian Science Fund (FWF) project W1255-N23, and the SecInt Doctoral College funded by TU Wien. This research was also funded in part by the Austrian Science Fund (FWF) [10.55776/F85, 10.55776/W1255], for open access purposes; the authors have applied a CC BY public copyright license to any author accepted manuscript version arising from this submission. We thank the authors of \cite{bayarmagnai2024algebraic} for alerting us to an error in an early version of our manuscript and the anonymous reviewers for their constructive feedback.


\bmhead{Data Availability}
The tool's source code, data, and benchmarks used in this work are available in \Tool's repository \url{https://github.com/probing-lab/polar}.

\bibliography{invariantsbib}

\end{document}

%% file: 0preamble.tex
\usepackage{graphicx}
\usepackage{mathtools}

\usepackage[usenames,dvipsnames,table]{xcolor}

\usepackage{subcaption}

  \usepackage{ifdraft}
  \usepackage[colorinlistoftodos,prependcaption,obeyFinal]{todonotes}

\usepackage{subcaption}

  \usepackage{ifdraft}

 \newcommand{\re}[1]{#1}
 \newcommand{\hl}[1]{#1}
\usepackage{lineno}
\usepackage{verbatim}
\usepackage{xfrac}
\usepackage{multirow}
\usepackage{booktabs}
\usepackage{dsfont}
\usepackage{xspace}
\usepackage{float}
\usepackage{tcolorbox}
\usepackage{lipsum}  
\usepackage{enumitem}
\usepackage{marvosym}
\usepackage{aliascnt}
\usepackage[nameinlink,capitalise]{cleveref}

\ifdraft{\linenumbers}{}

\usepackage{mleftright}
\usepackage{comment}
\usepackage[utf8]{inputenc}
\usepackage{tikz}
\usetikzlibrary{arrows,patterns,backgrounds,fit}

\renewcommand{\epsilon}{\ensuremath\varepsilon}

\newcommand{\R} {\mathbb{R}} 
 
\newcommand{\C} {\mathbb{C}} 
\newcommand{\N} {\mathbb{N}} 

\newcommand{\D} {\mathcal{D}} 
\newcommand{\Q} {\mathbb{Q}}
\newcommand{\E} {\mathbb{E}}
\newcommand{\cE} {\mathcal{E}}
\newcommand{\cR} {\mathcal{R}}

\renewcommand{\P} {\mathcal{P}}

\newcommand{\vars}[1]{\ensuremath{\operatorname{Vars}(#1)}\xspace}
\newcommand{\varsn}[2]{\ensuremath{\operatorname{Vars}_{#1}(#2)\xspace}}

\newcommand{\Tool}{{\texttt{Polar}}}

\DeclareMathOperator{\tr}{tr}
\DeclareMathOperator{\SL}{SL}

\newcommand{\bernoulli}[1]{\ensuremath{\operatorname{Bernoulli}({#1})}}
\newcommand{\Inv}[1]{\ensuremath{\operatorname{Inv}({#1})}}
\newcommand{\normal}[1]{\ensuremath{\operatorname{Normal}({#1})}}
\newcommand{\uniform}[1]{\ensuremath{\operatorname{Uniform}({#1})}}

\newcommand{\s}[1]{{S}(#1)}

\newcommand{\seq}[3][{}]{\langle #2 \rangle_{#3}^{#1}}

\DeclareMathAlphabet{\mathbbb}{U}{bbold}{m}{n}
\providecommand*{\eu}%
{\ensuremath{\mathrm{e}}}
\providecommand*{\iu}%
{\ensuremath{\mathrm{i}}}
\newcommand*\wc{{}\cdot{}}

\newcommand{\timeout}{\textsc{-}}

%% file: 01introduction.tex
\section{Introduction}

\subsection{Background and Motivation}

With substantial progress in computer-aided program analysis and automated reasoning, 
several techniques have emerged to automatically synthesise loop invariants, thus advancing
a central challenge in the computer-aided verification of programs with loops. 
In this paper, we 
address the problem of automatically generating loop
invariants in the presence of polynomial arithmetic, which is still unsolved.
This problem remains  unsolved even when we restrict consideration to loops that are non-nested, without conditionals, and/or without exit conditions. Our work improves the state of the art under such and similar considerations. 

\emph{Loop invariants}, in the sequel simply \emph{invariants}, are properties that hold before and after every iteration of a loop. Invariants therefore provide the key inductive arguments for automating the verification of programs; for example, proving correctness of deterministic loops~\cite{Rodriguez04,Kovacs08,Oliveira16,Kincaid18,Humenberger18} and correctness of hybrid and  probabilistic loops~\cite{Huang2014,KaminskiKMO16,BKS19}, or data flow analysis and compiler optimisation~\cite{muller2004computing}.
One challenging aspect in invariant synthesis is the derivation of \emph{polynomial invariants} for arithmetic loops.
Such invariants 
are defined  by polynomial relations 
\(P(x_1, \dots, x_k) = 0\) among the program variables $x_1, \ldots, x_k$. 
While  deriving polynomial invariants is, in general, undecidable~\cite{Ouaknine20}, efficient invariant synthesis techniques emerge when considering restricted classes of polynomial arithmetic in so-called \emph{solvable loops}~\cite{Rodriguez04}, such as loops with (blocks of) affine assignments~\cite{Kovacs08,Oliveira16,Humenberger18,Kincaid18}.

A common approach for constructing polynomial invariants, first pioneered in~\cite{Waldinger72,DBLP:journals/cacm/KatzM76}, is to (i) map a loop to a system of recurrence equations modelling the behaviour of program variables; (ii) derive closed-forms for program variables by solving the recurrences; and (iii) compute polynomial invariants by eliminating the loop counter $n$ from the closed-forms.
\re{The polynomial invariants resulting from step (iii) over-approximate the fixed point of the loop.}
The central components in this setting follow.
In step (i) a \emph{recurrence operator} is employed to map loops to recurrences, 
which leads to closed-forms for the program variables as \emph{exponential polynomials} in step (ii); that is, each program variable is written as a finite sum of the form $\sum_j P_j(n) \lambda_j^n$ parameterised by the \(n\)th loop iteration for polynomials $P_j$ and algebraic numbers $\lambda_j$. 
From the theory of algebraic recurrences, 
this is the case if and only if the behaviour of each variable obeys a linear recurrence equation with constant coefficients \cite{everest2003recurrence,kauers2011concrete}.
Exploiting this result, the class of recurrence operators that can be linearised are called \emph{solvable}~\cite{Rodriguez04}.
Intuitively, a loop with a recurrence operator is solvable only if the non-linear dependencies in the resulting system of polynomial recurrences are acyclic (see Section~\ref{section:recurrences}). 
However, even simple loops may fall outside the category of solvable operators, but still admit polynomial invariants and closed-forms for combinations of variables.
This phenomenon is illustrated in Figure~\ref{fig:running-examples} whose recurrence operators are not solvable (i.e. unsolvable).
\re{In other works, the generation of polynomial invariants is usually limited to those variables that admit closed-forms.
In our work, specifically step (iii), we can generate polynomial invariants from \emph{combinations of program variables} that admit closed-forms (where individual variables may fail to do so). This analysis can lead to a tighter over-approximation of a loop's fixed point.}
In general, the main obstacle in the setting of unsolvable recurrence operators is the absence of \lq\lq well-behaved" closed-forms for the resulting recurrences.

\begin{figure}[tb]
    \begin{subfigure}[c]{0.45\linewidth}
        \begin{algorithmic}
            \STATE \(z \leftarrow 0\)
            \WHILE {$\star$}
                \STATE \(z \leftarrow 1 - z\)
                \STATE \(x \leftarrow 2x + y^2 + z\)
                \STATE \(y \leftarrow 2y - y^2 + 2z\)
            \ENDWHILE
        \end{algorithmic}
        \begin{tcolorbox}[width=0.9\linewidth,boxrule=1pt,leftrule=2pt,arc=0pt,auto outer arc]
        \textbf{Closed-form of \(x+y\):} \\
        $x(n) + y(n) = 2^n(x(0) + y(0) + 2) - (-1)^n/2 - 3/2$
        \end{tcolorbox}
        \caption{The program $\P_\square$.}
        \label{fig:running:squares}
    \end{subfigure}
    \hfill
    \begin{subfigure}[c]{0.45\linewidth}
        \begin{algorithmic}
            \vspace{1em}
                \STATE \(x, y \leftarrow 1, 1\)
            \WHILE {$\star$}
                \STATE \(  w \leftarrow x + y \)
                \STATE \( x \leftarrow w^2 \)
                \STATE \( y \leftarrow w^3 \)
            \ENDWHILE
        \end{algorithmic}
        \begin{tcolorbox}[width=0.9\linewidth,boxrule=1pt,leftrule=2pt,arc=0pt,auto outer arc]
        \textbf{Polynomial Invariant:} \\
        $y^2(n)-x^3(n) = 0$
        \vspace{1em}
        \end{tcolorbox}
        \caption{The program $\P_\textrm{SC}$.}
        \label{fig:running:2}
    \end{subfigure}
        \caption{Two running examples with unsolvable recurrence operators. Nevertheless, $\P_\square$ admits a closed-form for combinations of variables and $\P_\textrm{SC}$ admits a polynomial invariant. Herein we use $\star$ (rather than a loop guard or \texttt{true}) as loop termination is not our focus.
        For the avoidance of doubt, in this paper we consider standard mathematical  arithmetic (e.g. mathematical integers) rather than machine floating-point and finite precision arithmetic.
        }
    \label{fig:running-examples}
\end{figure}

\subsection{Related Work} 
To the best of our knowledge, the study of invariant synthesis from the viewpoint of recurrence operators is mostly limited to the setting of solvable operators (or minor generalisations thereof).
In~\cite{Rodriguez04,Rodriguez07} the authors introduce solvable loops and mappings to model loops with (blocks of) affine assignments and propose solutions for steps (i)--(iii) for this class of loops: all polynomial invariants are derived by first solving linear recurrence equations and then eliminating variables based on Gröbner basis computation.
These results have further been generalised in~\cite{Kovacs08,Humenberger18} to handle more generic recurrences; in particular, deriving  arbitrary exponential polynomials as closed-forms of loop variables and allowing restricted multiplication among recursively updated loop variables.
The authors of~\cite{Farzan15,Kincaid18} generalise the setting: they consider more complex programs and devise abstract (wedge) domains to map the invariant generation problem to the problem of solving \emph{C-finite recurrences}. 
(We give further details of this class of recurrences in Section~\ref{sec:preliminaries}).
All the aforementioned approaches are mainly restricted to C-finite recurrences for which closed-forms always exist, thus yielding loop invariants. 
In \cite{BKS19,bns} the authors establish techniques to apply invariant synthesis techniques developed for deterministic loops to probabilistic programs.
Instead of devising recurrences describing the precise value of variables in step (i), their approach produces C-finite recurrences describing (higher) moments of program variables, yielding moment-based invariants after step (iii).

Pushing the boundaries in analyzing unsolvable loops is addressed in~\cite{Kincaid18,Frohn20}. The approach of~\cite{Kincaid18}  extracts C-finite recurrences over linear combinations of loops variables from unsolvable loops. For example, the method presented in~\cite{Kincaid18}  can also synthesise the closed-forms identified by our work for Figure~\ref{fig:running:squares}. However, 
unlike~\cite{Kincaid18}, our work is not limited to linear combinations (we can  extract C-finite recurrences over \emph{polynomial} relations in the loop variables). As such, the technique of~\cite{Kincaid18} cannot synthesise the polynomial loop invariant in Figure~\ref{fig:running:2}, whereas our work can.
A further related approach to our work is given in~\cite{Frohn20}, yet in the setting of loop termination. However, our work is not restricted to solvable loops that are  triangular,  but  can handle mutual dependencies among (unsolvable) loop variables, as evidenced in Figure~\ref{fig:running-examples}.

Related work in the literature introduces techniques from the theory of martingales in order to synthesise invariants in the setting of probabilistic programs \cite{chakarov2013probabilistic}.
Therein, the programming model is represented by a class of loop programs where all updates are linear and the synthesised invariants are given by linear templates.
By contrast, our method allows us to handle polynomial arithmetic; in particular, we automatically generate invariants given by monomials in the program variables.
On the other hand, the approach of~\cite{chakarov2013probabilistic} can also synthesise supermartingales whereas our  work is restricted to invariants defined by equalities.

\subsection{Our Contributions} 
In this paper we \hl{consider the sister problems of invariant generation and solvable loop synthesis} in the setting of unsolvable recurrence operators.
We introduce the notions of \emph{effective} and \emph{defective} program variables where, figuratively speaking, the defective variables are those \lq\lq responsible\rq\rq\ for unsolvability. Our main contributions are summarised below. 
    \begin{enumerate}
        \item Crucial for our synthesis technique is our  novel characterisation of unsolvable recurrence operators in terms of defective variables (Theorem~\ref{thm:def-char}).
        Our approach  complements existing techniques in loop analysis, by extending these methods to the setting of `unsolvable' loops.
        \item On the one hand, defective variables do not generally admit closed-forms.
                On the other hand, some polynomial combinations of such variables are well-behaved (see e.g., Figure~\ref{fig:running-examples}). 
                We show how to compute the set of defective variables in polynomial time (Algorithm~\ref{alg:nmcv-alg}).
        \item We introduce a new technique to synthesise valid linear relations in defective monomials such that these relations  admit closed-forms, from which polynomial loop invariants follow (Section~\ref{sec:invariants}).
        \item \hl{Given an unsolvable loop, we introduce an algorithmic approach   (Algorithm~\ref{alg:loop-synthesis}) that synthesises a solvable loop with the following property: every polynomial invariant of the solvable loop is also an invariant of the given unsolvable loop (Section~\ref{sec:loopsynthesis}).
        }
        \item We generalise our work to the analysis of probabilistic program loops (Section~\ref{sec:probabilistic}) and showcase further applications of  unsolvable operators in such programs (Section~\ref{section:case-studies}). 
        \item We provide a fully automated approach in the tool \Tool\footnote{\url{https://github.com/probing-lab/polar}}. \hl{For evaluating our work, we compiled an extensive lists of challenging loops from the literature, including applications of  mathematical and physical modelling}.  Our experiments demonstrate the feasibility of invariant synthesis for `unsolvable' loops and the applicability of our approach to deterministic loops, probabilistic models, and biological systems (Section~\ref{section:experiments}).
    \end{enumerate}

\subsection{Beyond Invariant Generation}
We believe our work can provide new solutions towards compiler optimisation challenges. \emph{Scalar evolution\footnote{\url{https://llvm.org/docs/Passes.html}}} is a technique to detect general induction variables.  
Scalar evolution and general induction variables are used for a multitude of compiler optimisations, for example inside the LLVM toolchain~\cite{LattnerA04}. 
On a high-level, general induction variables are loop variables that satisfy linear recurrences.
As we show in our work, defective variables do not satisfy linear recurrences in general;  hence,  scalar evolution optimisations cannot be applied upon them. 
However, some linear combinations of defective monomials \emph{do} satisfy linear recurrences, which opens avenues where we can apply scalar evolution techniques over such monomials. 
In particular, our work automatically computes 
  polynomial combinations of some defective loop variables, 
which potentially enlarges the class of loops that, for example, LLVM can optimise.

\subsection{Structure and Summary of Results} The rest of this paper is organised as follows. 
    We briefly recall preliminary material in Section~\ref{sec:preliminaries}. 
    Section~\ref{section:recurrences} abstracts from concrete recurrence-based approaches to invariant synthesis via recurrence operators. 
    Section~\ref{sec:effective-deffective} introduces effective and defective variables, presents Algorithm~\ref{alg:nmcv-alg} that computes the set of defective program variables in polynomial time, and characterises unsolvable loops in terms of defective variables (Theorem~\ref{thm:def-char}). 
   In Section~\ref{sec:invariants} we present our new technique that synthesises linear relations in defective monomials that admit well-behaved closed-forms. 
   \hl{In Section~\ref{sec:loopsynthesis} we introduce Algorithm~\ref{alg:loop-synthesis} to synthesise solvable loops.  That is, given an unsolvable loop, Algorithm~\ref{alg:loop-synthesis}  outputs a solvable loop, if it exists such that each polynomial invariant of the solvable loop is also an invariant of the unsolvable loop.}
   In Section~\ref{sec:probabilistic} we detail the necessary changes to the \hl{algorithms in Sections~\ref{sec:invariants} and \ref{sec:loopsynthesis}} for probabilistic programs.
    We illustrate our approach  with several case-studies in Section~\ref{section:case-studies}, and describe  a 
    fully-automated tool support of our work in Section~\ref{section:experiments}. We  also report on accompanying experimental evaluation in Sections~\ref{section:case-studies}--\ref{section:experiments}, and conclude the paper in Section~\ref{sec:conclusion}.

%% file: 02preliminaries.tex
\section{Preliminaries}
\label{sec:preliminaries}

 \subsection{Notation}
 Throughout this paper, we write $\N$, $\Q$, and $\R$ to respectively denote the sets of natural, rational, and real numbers.
 We write $\overline{\Q}$, the algebraic closure of $\Q$, to denote the field of algebraic numbers.
 We  write $\R[x_1,\ldots,x_k]$ and $\overline{\Q}[x_1,\ldots,x_k]$ for the polynomial rings of all polynomials $P(x_1, \ldots, x_k)$ in $k$ variables $x_1,\ldots,x_k$ with coefficients in $\R$ and  $\overline{\Q}$, respectively (with $k\in\N$ and $k\neq 0$). 
 A \emph{monomial} is a monic polynomial with a single term.
 
 For a program $\P$, $\vars\P$ denotes the set of program variables.
 We adopt the following syntax in our examples.
 Sequential assignments in while loops are listed on separate lines (as demonstrated in Figure~\ref{fig:running-examples}).
 In programs where simultaneous assignments are performed, we employ vector notation (as demonstrated by the assignments to the variables \(x\) and \(y\) in program \(\P_{\textrm{MC}}\) in Example~\ref{ex:NLMarkov-1}).

 We refer to a directed graph with $G$, whose edge and vertex (node) sets are respectively denoted via $A(G)$ and $V(G)$. 
 We endow each element of $A(G)$ with a label according to a labelling function $\mathcal{L}$. A \emph{path} in $G$ is a finite sequence of contiguous edges of $G$, whereas a \emph{cycle} in $G$  is a path whose initial and terminal vertices coincide.
 A graph that contains no cycles is \emph{acyclic}.
 In a graph \(G\), if there exists a path from vertex $u$ to vertex $v$, then
 we say that $v$ is \emph{reachable} from vertex $u$ and say that
 \(u\) is a \emph{predecessor} of \(v\).

\subsection{C-finite recurrences} 
We recall relevant results on (algebraic) recurrences  and refer to~\cite{everest2003recurrence,kauers2011concrete} for further details.
A \emph{sequence} in $\overline\Q$ is a function $u \colon\N\to\overline\Q$, shortly written also as  $\seq[\infty]{u(n)}{n=0}$ or simply just $\seq{u(n)}{n}$. A \emph{recurrence} for a {sequence} $\seq{u(n)}{n}$ is an equation 	$u(n{+}\ell) = \textrm{Rec}(u(n{+}\ell{-}1), \ldots, u(n{+}1), u(n),n)$, for some function $\textrm{Rec}\colon\R^{\ell+1}\to\R$.
The number $\ell\in\N$ is the \emph{order} of the recurrence.

A special class of recurrences we consider are the   
\emph{linear recurrences with constant coefficients}, in short \emph{C-finite recurrences}. 
A C-finite recurrence for a sequence \(\seq{u(n)}{n}\) is an equation of the form
\begin{equation}\label{eq:CFiniteRec}
	u(n{+}\ell) = a_{\ell{-}1} u(n{+}\ell{-}1) + a_{\ell{-}2} u(n{+}\ell{-}2) + \cdots + a_{0} u(n)
\end{equation}
where  $a_0,\ldots, a_{\ell-1}\in \overline\Q$ are constants and $a_0 \neq 0$.
A sequence \(\seq{u(n)}{n}\) satisfying a C-finite recurrence~\eqref{eq:CFiniteRec} is a \emph{C-finite sequence} and is uniquely determined by its initial values \(u_0=u(0),\ldots, u_{\ell-1}=u(\ell{-}1)\).  
The \emph{characteristic polynomial} associated with the  
C-finite recurrence relation~\eqref{eq:CFiniteRec} is
\begin{equation*}
    x^{n+\ell} - a_{\ell-1} x^{n+\ell-1} - a_{\ell-2} x^{n+\ell-2} - \cdots - a_{0} x^{n}.
\end{equation*}
The terms of a C-finite sequence can be written in a closed-form as exponential polynomials, depending only on $n$ and the initial values of the sequence. That is, if  $\seq{u(n)}{n}$ is determined by a C-finite recurrence~\eqref{eq:CFiniteRec}, then $u(n) = \sum_{k=1}^r P_k(n) \lambda_k^n$ where $P_k(n)\in\overline{\Q}[n]$ and $\lambda_1,\ldots, \lambda_r$ are  the roots of the associated characteristic polynomial. 
Importantly, closed-forms of (systems of) C-finite sequences always exist and are computable \cite{everest2003recurrence,kauers2011concrete}.

\subsection{Invariants} 
A loop invariant is a loop property that holds before and after each loop iteration~\cite{Hoare69}. In this paper, we are interested in \emph{polynomial invariants}, the class of invariants given by Boolean combinations of polynomial equations among loop variables. 
There is a minor caveat to our characterisation of (polynomial) loop invariants.
We assume that a (polynomial) invariant consists of a finite number of initial values
together with a closed-form expression of a monomial in the loop variables.
Thus the closed-form of a loop invariant must eventually hold after a (computable) finite number of loop iterations.
Let us illustrate this caveat with the following loop example. 
\begin{example}
\re{%
\begin{algorithmic}
\STATE \(x,y,z \leftarrow 0,1,0\)
\WHILE{$\star$}
\STATE \(x \leftarrow 1\)
\STATE \(y \leftarrow y+x\)
\STATE \(z \leftarrow z+1\)
\ENDWHILE
\end{algorithmic}
}
\re{The loop admits the polynomial invariant $y - z - 1 = 0$ given by the initial values
\(x(0)=0\), \(y(0)=1\), \(z(0)=0\) and the closed-forms \(x(n)=1\), \(y(n)=n+1\), and \(z(n)=n\).
For each \(n\ge 1\), we denote by $v(n)$ the value of a loop variable $v$ at loop iteration $n$.}
\end{example}

Herein, we synthesise invariants that satisfy inhomogeneous first-order recurrence relations and it is straightforward to show that each associated closed-form holds for \(n\ge 1\).

\subsection{Polynomial Invariants and Invariant Ideals}
A polynomial \emph{ideal} is a subset $I \subseteq \overline{\Q}[x_1,\ldots,x_k]$ with the following properties:
$I$ contains $0$; \(I\) is closed under addition; and if $P \in \overline{\Q}[x_1,\ldots,x_k]$ and $Q \in I$, then $PQ \in I$.
For a set of polynomials $S \subseteq \overline{\Q}[x_1,\ldots,x_k]$, one can define  the \emph{ideal generated by $S$} by
\begin{equation*}
    I(S) := \{ s_1 q_1 + \cdots + s_\ell q_\ell \mid s_i \in S, q_i \in \overline{\Q}[x_1,\ldots,x_k], \ell \in \N \}.
\end{equation*}

Let $\P$ be a program as before.
For  \(x_j\in\vars\P\), let \(\seq{x_j(n)}{n}\) denote the sequence whose
\(n\)th term is given by the value of \(x_j\) in the \(n\)th loop iteration.
The set of polynomial invariants of \(\P\) form an ideal, the \emph{invariant ideal} of \(\P\)~\cite{Rodriguez07}.
If for each program variable \(x_j\) the sequence \(\seq{x_j(n)}{n}\) is C-finite, then a basis for the invariant ideal can be computed as follows.
Let $f_j(n)$ be the exponential polynomial closed-form of variable $x_j$.
The exponential terms $\lambda_1^n, \ldots, \lambda_s^n$ in each of the $f_j(n)$ are replaced by fresh symbols, yielding the polynomials $g_j(n)$.
Next, with techniques from \cite{Kauers2008ComputingTA}, the set $R$ of all polynomial relations among $\lambda_1^n, \ldots, \lambda_s^n$ (that hold for each $n \in \N$) is computed.
Then we express the polynomial relations in terms of the fresh constants, so that we can interpret \(R\) as a set of polynomials.
Thus
\begin{equation*}
    I(\{ x_j - g_j(n) \mid 1 \leq i \leq k \} \cup R) \cap \overline{\Q}[x_1,\ldots,x_k]
\end{equation*}
is precisely the invariant ideal of \(\P\).
Finally,
we can compute a finite basis for the invariant ideal with techniques from Gröbner bases and elimination theory~\cite{Kauers2008ComputingTA}.

%% file: 03recurrences.tex
\section{From Loops to Recurrences}
\label{section:recurrences}

\subsection{Recurrence Operators}

Modelling properties of loop variables by algebraic recurrences and solving the resulting recurrences is an established approach in program analysis.
Multiple works~\cite{Kovacs08,Farzan15,Kincaid18,Humenberger17,Humenberger18}  associate a loop variable $x$ with a sequence $\seq{x(n)}{n}$  whose \(n\)th term is given by the value of $x$ in the $n$th loop iteration.
These works are primarily concerned with the problem of representing such sequences via recurrence equations whose closed-forms can be computed automatically, as in the case of C-finite sequences. 
A closely connected question to this line of research focuses on identifying  classes of loops that can be modelled by solvable recurrences, as advocated in~\cite{Rodriguez04}. 
To this end,  over-approximation methods for general loops are proposed in~\cite{Farzan15,Kincaid18}  
 such that solvable recurrences can be obtained from (over-approximated) loops. 

In order to formalise the above and similar efforts in  associating  loop variables with recurrences, herein we introduce the concept of a \emph{recurrence operator}, and the characterisation of both \emph{solvable} and \emph{unsolvable operators}. 
Intuitively, a recurrence operator maps program variables to recurrence equations describing some properties of the variables; for instance, the exact values at the $n$th loop iteration~\cite{Rodriguez04,Kovacs08,Farzan15} or statistical moments in probabilistic loops~\cite{BKS19}.

\medskip

\begin{definition}[Recurrence Operator]
A \emph{recurrence operator} $\cR$ maps the program variables $\vars\P$ to the polynomial ring $\R[\operatorname{Vars}_n(\P)]$.
The set of equations 
    \begin{equation*}
        \{ x(n{+}1) = \cR[x] \mid x \in \vars\P \}
    \end{equation*} 
constitutes a polynomial first-order system of recurrences.
We call $\cR$ \emph{linear} if $\cR[x]$ is linear for all $x \in \vars\P$.

One can extend the operator $\cR$ to $\R[\vars\P]$.
Then, with a slight abuse of notation, for $P(x_1, \ldots, x_j) \in \R[\vars\P]$ we define $\cR(P)$ by $P(\cR[x_1], \ldots, \cR[x_j])$.
\end{definition}

\medskip

For a program $\P$ with recurrence operator $\cR$ and a monomial over program variables $M := \prod_{x \in \vars\P} x^{\alpha_x}$, we denote by $M(n)$ the product of sequences $\prod_{x \in \vars\P} x^{\alpha_x}(n)$.
Given a polynomial $P$ over program variables, $P(n)$ is defined by replacing every monomial $M$ in $P$ by $M(n)$.
For a set $T$ of polynomials over program variables let $T_n := \{ P(n) \mid P \in T \}$.

\medskip

\begin{example}\label{remark:running2_recs}
Consider the program \(\P_\textrm{SC}\) in Figure~\ref{fig:running:2}.
One can employ a recurrence operator $\cR$ in order to capture the values of the program variables in the $n$th iteration.
For \(v\in\vars{\P_\textrm{SC}}\), \(\cR[v]\) is obtained by bottom-up substitution in the polynomial updates starting with \(v\). As a result, we obtain the following system of recurrences:
\begin{align*}
    w(n{+}1) &= \cR[w] = x(n) + y(n) \\
    x(n{+}1) &= \cR[x] = x(n)^2 + 2x(n)y(n) + y(n)^2\\
    y(n{+}1) &= \cR[y] = x(n)^3 + 3x(n)^2y(n) + 3x(n)y(n)^2 + y(n)^3.
\end{align*}

Similarly, for the program $\P_\square$ of Figure~\ref{fig:running:squares}, we obtain the following system of recurrences:  
\begin{align*}
    z(n{+}1) &= \cR[z] = 1 - z(n) \\
    x(n{+}1) &= \cR[x] = 2x(n) + y(n)^2 - z(n) + 1 \\
    y(n{+}1) &= \cR[y] = 2y(n) - y(n)^2 - 2z(n) + 2.
\end{align*}
\end{example}

\subsection{Solvable Operators} 
Systems of linear recurrences with constant coefficients admit computable closed-form solutions as exponential polynomials~\cite{everest2003recurrence,kauers2011concrete}.
This property holds for a larger class of recurrences with polynomial updates, which leads to the notion of \emph{solvability} introduced in \cite{Rodriguez04}.
We adjust the notion of solvability to our setting by using recurrence operators. 
In the following definition, we make a slight abuse of notation and order the program variables so that we can transform program variables by a matrix operator.

\medskip

\begin{definition}[Solvable Operators \cite{Rodriguez04,Oliveira16}]
\label{def:solvability}
The recurrence operator $\cR$ is \emph{solvable} if there exists a partition of $\operatorname{Vars}_n$; that is, $\operatorname{Vars}_n = W_1 \uplus \dots \uplus W_k$ such that for $x(n) \in W_j$,
\begin{equation*}
    \cR[x] = M_j \cdot W_j^\top + P_j(W_1,\dots,W_{j-1})
\end{equation*}
for some matrices $M_j$ and polynomials $P_j$.
A recurrence operator that is not solvable is said to be \emph{unsolvable}.
\end{definition}

\medskip 
This definition captures the notion of solvability in \cite{Rodriguez04} (see the discussion in~\cite{Oliveira16}).

We conclude this section by emphasising the use of (solvable) recurrence operators beyond deterministic loops, in particular relating its use to probabilistic program loops. As evidenced in~\cite{BKS19}, 
recurrence operators  model statistical moments of program variables by essentially focusing on  {solvable} recurrence operators extended with an expectation operator $\E(\wc)$ to derive closed-forms of (higher) moments of program variables, as illustrated below. 

\medskip

\begin{example}
\label{ex:NLMarkov-1}
Consider the probabilistic program $\P_\textrm{MC}$ of ~\cite{schreuder2019invariants,chakarov2016deductive} modelling a non-linear Markov chain, where $\bernoulli{p}$ refers to a Bernoulli distribution with parameter $p$.
Here the updates to the program variables \(x\) and \(y\) occur simultaneously.

\begin{algorithmic}
\WHILE{$\star$}
\STATE \(s \leftarrow \bernoulli{1/2}\)
\IF{\(s=0\)}
\STATE \(   
\begin{pmatrix}
x\\
y
\end{pmatrix}
\leftarrow
\begin{pmatrix}
x + xy \\
\frac{1}{3}x + \frac{2}{3}y + xy
\end{pmatrix}
\)
\ELSE
\STATE \(   
\begin{pmatrix}
x\\
y
\end{pmatrix}
\leftarrow
\begin{pmatrix}
x + y +  \frac{2}{3}xy \\
2y + \frac{2}{3}xy
\end{pmatrix}
\)
\ENDIF
\ENDWHILE
\end{algorithmic}

We can construct recurrence equations, in terms of the expectation operator \(\E(\wc)\), for this program as follows:
\begin{align*}
    \E(s_{n+1}) &= \tfrac{1}{2} \\
    \E(x_{n+1}) &= \E(x_n) + \tfrac{1}{2} \E(y_n) + \tfrac{5}{6} \E(x_n y_n) \\
    \E(y_{n+1}) &= \tfrac{1}{6}\E(x_n) + \tfrac{4}{3} \E(y_n) + \tfrac{5}{6} \E(x_n y_n).
\end{align*}
\end{example}

%% file: 04defective.tex
\section{Defective Variables}
\label{sec:effective-deffective}

To the best of our knowledge, existing approaches in loop analysis and invariant synthesis are restricted to solvable recurrence operators. 
In this section, we establish a new characterisation of unsolvable recurrence operators.
Our characterisation pinpoints the program variables responsible for unsolvability, the \emph{defective variables} (see Definition~\ref{def:effective}).
Moreover, we provide a polynomial time algorithm to compute the set of defective variables (Algorithm~\ref{alg:nmcv-alg}), in order to exploit our new characterisation for synthesising invariants in the presence of unsolvable operators in Section~\ref{sec:invariants}.

For simplicity, we limit the discussion in this section to deterministic programs.
We note however that the results presented herein can also be applied to probabilistic programs.
The details of the necessary changes in this respect  are given in Section~\ref{sec:probabilistic}.

In what follows, we write  $\mathcal{M}_n(\P)$ to denote the set of non-trivial \emph{monomials in $\vars\P$ evaluated at the $n$th loop iteration} so that
\begin{equation*}
\mathcal{M}_n(\P) := \mleft\{ \textstyle{\prod_{x \in \vars\P}} x^{\alpha_x}(n) \mid \exists x\in\vars\P \text{ with } \alpha_x \neq 0 \mright\}.
\end{equation*}
We next  introduce the notions of variable dependency and dependency graph, needed to further characterise defective variables. 

\medskip

\begin{definition}[Variable Dependency]\label{def:dependency}
Let $\P$ be a loop with recurrence operator $\cR$ and $x,y \in \vars{\P}$.
We say $x$ \emph{depends on} $y$ if $y$ appears in a monomial in $\cR[x]$ with non-zero coefficient.
Moreover, $x$ \emph{depends linearly on} $y$ if all monomials with non-zero coefficients in $\cR[x]$ containing $y$ are linear.
Analogously, $x$ \emph{depends non-linearly on} $y$ if there is a non-linear monomial with non-zero coefficient in $\cR[x]$ containing $y$.
	
Furthermore, we consider the transitive closure for variable dependency. 
If $z$ depends on $y$ and $y$ depends on $x$, then $z$ depends on $x$ and, if in addition, one of these two dependencies is non-linear, then $z$ depends non-linearly on $x$.
We otherwise say the dependency is linear.
\end{definition}

\medskip

For each program with polynomial updates, we further define a \emph{dependency graph} with respect to a recurrence operator.

\medskip

\begin{definition}[Dependency Graph]
\label{def:dependency-graph}
Let $\P$ be a program with recurrence operator $\cR$.
The \emph{dependency graph} of $\P$ with respect to $\cR$ is the labelled directed graph $G=(\vars\P, A, \mathcal{L})$ with vertex set \(\vars{\P}\), edge set
$A := \{ (x, y) \mid x, y \in \vars\P \land x \text{ depends on } y \}$, and a function $\mathcal{L} \colon A \to \{ L, N \}$ that assigns a unique \emph{label} to each edge such that
\begin{equation*}
    \mathcal{L}(x, y) = 
    \begin{cases} 
          L & \text{if } x \text{ depends \emph{linearly} on } y, \text{ and} \\
          N & \text{if } x \text{ depends \emph{non-linearly} on } y.
    \end{cases}
\end{equation*}
\end{definition}
\medskip

\re{Given a program and a recurrence operator, its dependency graph can be constructed automatically with standard techniques.}
In our approach, we partition the variables $\vars\P$ of the program $\P$ into two sets:  \emph{effective-} and \emph{defective variables},  denoted by $E(\P)$ and $D(\P)$ respectively.
Our partition builds on the definition of the dependency graph of \(\P\), as follows. 

\medskip

\begin{definition}[Effective-Defective] \label{def:effective}
A variable $x \in \vars\P$ is \emph{effective} if:
\begin{enumerate}
    \item \label{cond1} $x$ appears in no directed cycle with at least one edge with an $N$ label, and
    \item \label{cond2} $x$ cannot reach a vertex of an aforementioned cycle (as in \ref{cond1}).
\end{enumerate}
A variable is \emph{defective} if it is not effective.
\end{definition}

\medskip

\begin{example}     
From the recurrence equations of Example~\ref{remark:running2_recs} for the program $\P_\textrm{SC}$ (see Figure~\ref{fig:running:2}), one obtains the dependencies between the program variables of \(\P_\textrm{SC}\): 
the program variable $w$ depends linearly on both $x$ and $y$, whilst $x$ and $y$ depend non-linearly on each other and on $w$.
By definition, the partition into effective and defective variables is $E(\P_\textrm{SC}) = \emptyset$ and $D(\P_\textrm{SC}) = \{w, x, y\}$. 

 Similarly, we can construct the dependency graph for the program $\P_\square$ from Figure~\ref{fig:running:squares}, as illustrated in Figure~\ref{fig:TC:example}.
We derive  that $E(\P_\square) = \{z\}$ and $D(\P_\square) = \{x, y\}$.
\begin{figure}[htb]
 \centering
 \begin{minipage}[t]{0.32\textwidth}\vspace{0pt}
 \vspace{2.1em}
 \centering
 \resizebox{0.85\textwidth}{!}{%
     \begin{tikzpicture}[node distance={22mm}, thick, main/.style = {draw, circle}] 
         \node[main] (w) {$w$}; 
         \node[main] (x) [below left of=w] {$x$}; 
         \node[main] (y) [below right of=w] {$y$}; 
         \draw[->] (w) edge [bend right=8] node[left] {$L$} (x);
         \draw[->] (x) edge [bend right=8] node[right] {$N$} (w);
         \draw[->] (w) edge [bend right=8] node[left] {$L$} (y);
         \draw[->] (y) edge [bend right=8] node[right] {$N$} (w);
         \draw[->] (x) edge [bend right=8] node[inner sep=2pt,below] {$N$} (y);
         \draw[->] (y) edge [bend right=8] node[inner sep=2pt,above] {$N$} (x);
     \end{tikzpicture}
 }
 \end{minipage}
 \quad
 \begin{minipage}[t]{0.32\textwidth}\vspace{0pt}
 \centering
 \resizebox{0.85\textwidth}{!}{%
\begin{tikzpicture}[node distance={22mm}, thick, main/.style = {draw, circle}] 
        \node[main] (z) {$z$}; 
        \node[main] (x) [below left of=z] {$x$}; 
        \node[main] (y) [below right of=z] {$y$}; 
        \path (z) edge [loop above] node {L} (z);
        \draw[->] (x) edge [bend left=8] node[left] {$L$} (z);
        \draw[->] (y) edge [bend right=8] node[right] {$L$} (z);
        \draw[->] (x) edge [bend right=8] node[above] {$N$} (y);
        \path (x) edge [loop below] node {$L$} (x);
        \path (y) edge [loop below] node {$N$} (y);
    \end{tikzpicture}
}
\end{minipage}
 \caption{The dependency graphs for \(\P_\textrm{SC}\) and $\P_\square$ from Figure~\ref{fig:running-examples}.}
 \label{fig:TC:example}
 \end{figure}
\end{example}

\hl{
We give the following straightforward corollary of Definition~\ref{def:effective}.
\medskip
\begin{corollary} \label{cor:effective}
    Given any effective variable \(x\in E(\P)\), the recurrence relation \(\cR[x]\) is a polynomial in effective variables.
\end{corollary}
}

\medskip

The concept of effective, and, especially, defective variables allows us to establish a new characterisation of programs with unsolvable recurrence operators: {\it a recurrence operator is unsolvable if and only if there exists a defective variable} (as stated in Theorem~\ref{thm:def-char} and automated in Algorithm~\ref{alg:nmcv-alg}). We formalise and prove this results via the following three lemmas. 

\medskip

\begin{lemma}\label{lem:unsolvable:1}
Let $\P$ be a program with recurrence operator $\cR$.
If \(D(\P)\) is non-empty,  so that there is at least one defective variable, then $\cR$ is unsolvable.
\end{lemma}

\medskip
\begin{proof}
Let $x \in \vars{\P}$ be a defective variable and $G=(\vars\P, A, \mathcal{L})$ the dependency graph of \(\P\) with respect to a recurrence operator \(\cR\).
Following Definition~\ref{def:effective}, there exists a cycle \(C\) such that $x$ is a vertex visited by or can reach said cycle and, in addition, there is an edge in \(C\) labelled by \(N\).

Assume, for a contradiction, that $\cR$ is solvable.
Then there exists a partition $W_1,\ldots,W_k$ of $\varsn{n}\P$ as described in Definition~\ref{def:solvability}.
Moreover, since \(C\) is a cycle, there exists \(j\in\{1,\ldots, k\}\) such that each variable visited by \(C\) lies in \(W_j\).
Let \((w,y)\in C\) be an edge labelled with \(N\).
Since \(w\) depends on \(y\) non-linearly, and 
        \(\cR[w] = M_j \cdot W_j^\top + P_j(W_1,\ldots, W_{j-1})\)
(by Definition~\ref{def:solvability}),
it is clear that \(y(n) \in W_\ell\) for some \(\ell \neq j\).
We also have that \(y(n) \in W_j\) since \(C\) visits \(y\).
Thus we arrive at a contradiction as $W_1,\ldots,W_k$ is a partition of $\varsn{n}\P$.
Hence \(\cR\) is unsolvable.
\end{proof}

Given a program \(\P\) whose variables are all effective, it is immediate that a pair of distinct mutually dependent variables are necessarily linearly dependent and, similarly, a self-dependent variable is necessarily linearly dependent on itself.
Consider the following binary relation \(\sim\) on program variables:
\begin{equation*}
        x \sim y \iff x = y \lor (x \text{ depends on } y \land y \text{ depends on } x).
\end{equation*}
Thus, any two mutually dependent variables are related by \(\sim\).
Under the assumption that all variables of a program \(\P\) are effective, it is easily seen that \(\sim\) defines an equivalence relation on \(\vars\P\).
The partition of the equivalence classes \(\Pi\) of \(\vars\P\) under \(\sim\) admits the following notion of dependence between equivalence classes: for \(\pi,\hat{\pi}\in \Pi\) we say that
\(\pi\) \emph{depends on} \(\hat{\pi}\) if there exist variables \(x\in\pi\) and \(y\in\hat{\pi}\) such that variable \(x\) depends on variable \(y\).

\medskip

\begin{lemma}\label{lem:equivalencegraph}
Suppose that all variables of a program \(\P\) are effective.
Consider the graph \(\mathcal{G}\) with vertex set given by the set of equivalence classes \(\Pi\) and edge set \(A' := \{ (\pi, \hat{\pi}) \mid (\pi \neq \hat{\pi}) \land (\pi \text{ depends on } \hat{\pi}) \}\).
Then $\mathcal{G}$ is acyclic.
\end{lemma}
\begin{proof}
From the definition of \(\mathcal{G}\), it is clear that the graph is directed and has no self-loops.
Now assume, for a contradiction, that \(\mathcal{G}\) contains a cycle.
Since the relation \(\sim\) is transitive, there exists a cycle $C$ in \(\mathcal{G}\) of length two.
Moreover, the variables in a given equivalence class are mutually dependent.
Thus the elements of the two classes in \(C\) are equivalent under the relation \(\sim\), which contradicts the partition into distinct equivalence classes.
Therefore the graph \(\mathcal{G}\) is acyclic, as required.
\end{proof}

\begin{lemma}\label{lem:unsolvable:2}
Let $\P$ be a program with recurrence operator $\cR$.
If each of the program variables of \(\P\) is effective then $\cR$ is solvable.
\end{lemma}
\begin{proof}
By Lemma~\ref{lem:equivalencegraph}, the associated graph $\mathcal{G} = (\Pi, A')$ on the equivalence classes of \(\vars\P\) is directed and acyclic.
Thus there exists a topological ordering of $\Pi = \{ \pi_1, \dots, \pi_{|\Pi|} \}$ such that for every $(\pi_i, \pi_j) \in A'$ we have $i > j$.
Thus if $x \in \pi_i$ then $x$ does not depend on any variables in class $\pi_j$ for $j > i$.
Moreover, for each $\pi_i \in \Pi$, if $x, y \in \pi_i$ then $x$ cannot depend on $y$ non-linearly because every variable is effective (and all the variables in $\pi_i$ are mutually dependent).
Thus $\Pi$ evaluated at loop iteration $n$ partitions $\varsn{n}{\P}$ and satisfies the criteria in Definition~\ref{def:solvability}.
We thus conclude that $\cR$ is solvable.
\end{proof}

Together, Lemmas~\ref{lem:unsolvable:1}--\ref{lem:unsolvable:2} yield a new characterisation of unsolvable operators.

\medskip 

\begin{theorem}[Defective Characterisation]\label{thm:def-char}
Let $\P$ be a program with recurrence operator $\cR$, then $\cR$ is unsolvable if and only if \(D(\P)\) is non-empty.
\end{theorem}

\medskip

In Algorithm~\ref{alg:nmcv-alg} we provide a polynomial time algorithm that constructs both $E(\mathcal{P})$ and $D(\mathcal{P})$ given a program and a recurrence operator.
We use the initialism ``DFS" for the \emph{depth-first search} procedure.
Algorithm~\ref{alg:nmcv-alg} terminates in polynomial time as both the construction of the dependency graph and depth-first search exhibit polynomial time complexity.
The procedure searches for cycles in the dependency graph with at least one non-linear edge (labelled by \(N\)).
All variables that reach such cycles are, by definition, defective.

\begin{algorithm}[t]
\caption{Construct $E(\P)$ and $D(\P)$ from program $\P$ with operator $\cR$.}
\begin{algorithmic}
    \STATE \re{Construct the dependency graph $G = (\vars\P, A, \mathcal{L})$ of $\P$ with respect to $\cR$.}
    \STATE $D(\mathcal{P}) \leftarrow \emptyset$
    \FOR{$(x, y) \in A$ where $\mathcal{L}(x,y) = N$} 
        \IF {$x = y$}
            \STATE $\texttt{predecessor} \leftarrow \emptyset$
            \STATE $\text{DFS}(x, \texttt{predecessor})$
            \STATE $D(\mathcal{P}) \leftarrow D(\mathcal{P}) \cup \texttt{predecessor}$
        \ENDIF
        \IF {$x \neq y$}
            \STATE $\texttt{predecessor} \leftarrow \emptyset$
            \STATE $\text{DFS}(y, \texttt{predecessor})$
            \IF {$x \in \texttt{predecessor}$}
                \STATE $D(\mathcal{P}) \leftarrow D(\mathcal{P}) \cup \texttt{predecessor}$
            \ENDIF
        \ENDIF
    \ENDFOR
    \STATE $E(\mathcal{P}) \leftarrow \vars\P \setminus D(\mathcal{P})$
\end{algorithmic}
\label{alg:nmcv-alg}
\end{algorithm}

In what follows, we focus on programs with unsolvable recurrence operators, or equivalently by Theorem~\ref{thm:def-char}, the case where \(\D(\P) \neq \emptyset \).
The characterisation of unsolvable operators in terms of defective variables and our polynomial algorithm to construct the set of defective variables is the foundation for our approach synthesising invariants in the presence of unsolvable recurrence operators in Section~\ref{sec:invariants}.

\medskip 

\begin{remark}
The recurrence operator \(\mathcal{R}[x]\) for an effective variable \(x\) will admit a closed-form solution for every initial value \(x_0\).
For the avoidance of doubt, the same cannot be said for the recurrence operator of a defective variable.
However, it is possible that a set of initial values will lead to a closed-form expression as a C-finite sequence: consider a loop with defective variable \(x\) and update 
 \(x\leftarrow x^2\) and initialisation \(x_0 \leftarrow 0\) or \(x_0\leftarrow \pm 1\).
\end{remark}

%% file: 05invariants.tex
\section{Synthesising Invariants}
\label{sec:invariants}

In this section we propose a new technique to {\it synthesise invariants for programs with unsolvable recurrence operators}.
The approach is based on our new characterisation of unsolvable operators in terms of defective monomials (Section~\ref{sec:effective-deffective}).

For the remainder of this section we fix a program $\P$ with an unsolvable recurrence operator $\cR$, or equivalently with $D(\P) \neq \emptyset$.
We start by extending the notions of \emph{effective} and \emph{defective} from program variables to monomials of program variables.
Let $\mathcal{E}$ be the set of \emph{effective monomials} given by
\[
    \mathcal{E}(\P) = \mleft\{ \prod_{x \in E(\P)} x^{\alpha_x} \mid \alpha_x \in \N \mright\}.
\]
The complement, the \emph{defective monomials}, is given by 
$\mathcal{D}(\P) := \mathcal{M}(\P) \setminus \mathcal{E}(\P)$.
The difficulty with defective variables is that in general they do not admit closed-forms.
However, \hl{linear combinations of defective monomials} may allow for closed-forms as illustrated in previous examples.
The main idea of our technique for invariant synthesis in the presence of defective variables is to find such polynomials.
We fix a \emph{candidate polynomial} called $\s{n}$ based on an arbitrary degree $d \in \N$:

\begin{equation}
    \label{eqn:sn-eq}
   \s{n} = \sum_{W \in \D_n(\P)\restriction_d} c_W W, 
\end{equation}
where the coefficients $c_W\in\R$ are unknown real constants.
We use \(\D_n(\P)\restriction_d\) to indicate the set of \emph{defective monomials of degree at most $d$}.

\medskip
 
\begin{example} For $\P_\square$ in Figure~\ref{fig:running:squares} we have
\(\D_n(\P_\square)\restriction_1 = \{x,y\}\),
and \(\D_n(\P_\square)\restriction_2 = \{x,y,x^2, y^2, xy, xz, yz\}\).
\end{example}

\medskip

On the one hand, all variables in $\s{n}$ are defective; however, $\s{n}$ may admit a closed-form.
This occurs if \(\s{n}\) obeys a ``well-behaved" recurrence equation; that is to say,
an inhomogeneous recurrence equation where the inhomogeneous component is given by a linear combination of effective monomials.
In such instances the recurrence takes the form
\begin{equation}
    \label{eqn:main-rec}
    \s{n{+}1} = \kappa\s{n} + \sum_{M \in \cE_{n}(\P)} c_M M
\end{equation}

where the coefficients $c_M$ are unknown.
Thus an intermediate step towards our goal of synthesising invariants is to determine whether there are constants $c_M, c_W, \kappa \in \R$ that satisfy the above equations.
If such constants exist then we come to our final step: solving a first-order inhomogeneous recurrence relation.
There are standard methods available to solve first-order inhomogeneous recurrences of the form $\s{n{+}1} = \kappa \s{n} + h(n)$, where $h(n)$ is the closed-form of $\sum_{M \in \cE_{n}(\P)}c_M M$, see e.g.,~\cite{kauers2011concrete}.
We note \(h(n)\) is computable and an exponential polynomial since it is determined by a linear sum of effective monomials. 
Thus $\seq{S(n)}{n}$ is a C-finite sequence.

\medskip

\begin{remark}
Observe that the sum on the right-hand side of equation~\eqref{eqn:main-rec} is finite, since  all but finitely many of the coefficients \(c_M\) are zero.
Further, the coefficient $c_M$ of monomial $M$ is non-zero only if $M$ appears in $\cR[S]$.
\end{remark}

\medskip

Going further, in equation~\eqref{eqn:main-rec} we express $\s{n{+}1}$ in terms of a polynomial in $\varsn{n}\P$ with unknown coefficients $c_M,c_W$, and $\kappa$.
An alternative expression for $\s{n{+}1}$ in $\varsn{n}\P$ is given by the recurrence operator $\s{n{+}1} = \cR[S]$.
Taken in combination, we arrive at the following formula
\begin{equation*}
    \cR[S] - \kappa\s{n} - \sum_{M \in \cE_{n}(\P)} c_M M = 0, 
\end{equation*}
yielding a polynomial in $\varsn{n}\P$.
Thus all the coefficients in the above formula are necessarily zero as the polynomial is identically zero.
Therefore \emph{all} solutions to the unknowns $c_M, c_W$, and $\kappa$ are computed by solving a (quadratic) system of equations.
\re{
The main complexity of our invariant synthesis technique lies in solving the quadratic system.
In the candidate polynomial, every monomial in defective variables (of degree at most $d$) is associated with a unique unknown coefficient.
Hence, the size of the quadratic system can be polynomial in $d$ and exponential in the number of defective variables.
}

\medskip

\begin{example}\label{ex:squares_solved}
Consider the following illustration of our invariant synthesis procedure.
Recall program $\P_\square$ from Figure~\ref{fig:running:squares}:
\begin{algorithmic}
\STATE \(z \leftarrow 0\)
\WHILE {$\star$}
\STATE \(z \leftarrow 1 - z\)
\STATE \(x \leftarrow 2x + y^2 + z\)
\STATE \(y \leftarrow 2y - y^2 + 2z\)
\ENDWHILE
\end{algorithmic}
From Algorithm~\ref{alg:nmcv-alg} we obtain \(E(\P_\square) = \{z\}\) and \(D(\P_\square) = \{x, y\}\).
Because $D(\P_\square) \neq \emptyset$, we deduce using Theorem~\ref{thm:def-char} that the associated operator \(\cR\) is unsolvable.
Consider the {candidate} \(\s{n} = ax(n) + by(n)\) with unknowns \(a,b\in\R\).
The recurrence for $\s{n}$ given by $\cR$ is
\begin{multline*}
    \s{n{+}1} = \cR[S] = a\cR[x] + b\cR[y] \\ = a + 2b +2ax(n) + 2by(n) -(a+2b)z(n) + (a-b)y^2(n).
\end{multline*}

We next express $\s{n{+}1}$ in terms of an inhomogeneous recurrence equation (cf. equation~\eqref{eqn:main-rec}).
When we substitute for \(\s{n}\), we obtain
\begin{equation*}
\s{n{+}1} = \kappa (ax(n) + by(n)) + (cz(n) + d)
\end{equation*}
where the coefficients in the inhomogeneous component are unknown.
We then combine the preceding two equations (for brevity we suppress the loop counter \(n\) in the program variables \(x,y,z\)) and derive 
\begin{equation*}
(a + 2b - d) + (-a - c - 2b)z + (2a - \kappa a)x + (2b - \kappa b)y + (a - b)y^2 = 0.
\end{equation*}
Thus we have a polynomial in the program variables that is identically zero.
Therefore, all the coefficients in the above equation are necessarily zero.
We then solve the resulting system of quadratic equations, which leads to the non-trivial solution \(a=b\), \(\kappa=2\), \(d=3a\), and \(c=-3a\).
We substitute this solution back into the recurrence for \(\cR[S]\) and find
\begin{equation*}
    \s{n{+}1} = 2\s{n} + 3a(1-z(n)) = 2\s{n} + 3a \frac{1 + (-1)^n}{2}.
\end{equation*}
Here, we have used the closed-form solution \(z(n) = {1}/{2} - {(-1)^n}/{2}\) of the effective variable \(z\).
We can compute the solution of this inhomogeneous first-order recurrence equation.
In the case that \(a=1\),
we have \(\s{n} = 2^n (\s{0} + 2) - {(-1)^n}/{2} -{3}/{2}\). 
Therefore, the following identity holds for each $n \in \N$:
\begin{equation*}
    x(n) + y(n) = 2^n(x(0) + y(0) + 2) - {(-1)^n}/{2} - {3}/{2}
\end{equation*}
and so we have synthesised the closed-form of \(x+y\) for program \(\P_\square\) of Figure~\ref{fig:running:squares}.
\end{example}

\subsection{Solution Space of Invariants for Unsolvable Operators}

Given a program and a recurrence operator, our invariant synthesis technique is relative-complete with respect to the degree $d$ of the candidate $\s{n}$.
This means, for a fixed degree $d \in \N$, our approach is in theory able to compute \emph{all} polynomials of defective variables with maximum degree $d$ that satisfy a ``well-behaved" recurrence; that is, a first-order recurrence equation of the form \eqref{eqn:main-rec}.
This holds because of our reduction of the problem to a system of quadratic equations for which all solutions are computable.
\re{It is not guaranteed that a solution does exist.}
\re{In that case,} our technique can rule out the existence of well-behaved polynomials of defective variables of degree at most $d$ if the resulting system has no (non-trivial) solutions.

\re{
\begin{example}
The following loop models the \emph{logistic map}~\cite{May1976} which is well-known for its chaotic behaviour.
\begin{algorithmic}
\WHILE {$\star$}
\STATE \(x \leftarrow rx(1-x)\)
\ENDWHILE
\end{algorithmic}
The single variable $x$ is defective due to its non-linear self-dependency.
For most values of $r$ and initial values of $x$ the logistic map does not admit an analytical solution and well-behaved polynomials in $x$ do not exist.~\cite{Maritz2020}.
Hence, our invariant synthesis technique provides no solution for candidates of fixed degrees.
\end{example}
}

Let $\P$ be a program with program variables $\vars\P = \{x_1, \dots, x_k \}$.
The set of polynomials $P$ with $P(x_1(n), \dots, x_k(n)){=}0$ for all $n \in \N$ form an ideal, the \emph{invariant ideal} of $\P$.
The requirement of closed-forms is the main obstacle for computing a basis for the invariant ideal in the presence of defective variables.
Our work introduces a method that includes defective variables in the computation of invariant ideals, via the following steps of deriving the {\it polynomial invariant ideal of an unsolvable loop:}
\begin{itemize}
\item For every effective variable $x_i$, let $f_i(n)$ be its closed-form and assume $h(n)$ is the closed-form for some candidate $S$ given by a  polynomial in defective variables.
\item Let $\lambda_1^n, \ldots, \lambda_s^n$ be the exponential terms in all $f_i(n)$ and $h(n)$.
Replace the exponential terms in all $f_i(n)$ as well as $h(n)$ by fresh constants to construct the polynomials $g_i(n)$ and $l(n)$ respectively.
\item  Next, construct the set $R$ of polynomial relations among all exponential terms, as explained in Section~\ref{sec:preliminaries}.
Then, the ideal
\begin{equation*}
     I(\{ x_i - g_i(n) \mid x_i \in E(\P) \} \cup \{ S - l(n) \} \cup R) \cap \overline{\Q}[x_1, \dots, x_k]
\end{equation*}
contains precisely all polynomial relations among program variables implied by the equations $\{ x_i = f_i(n) \} \cup \{ S = g(n) \}$ in the theory of polynomial arithmetic.
\item A finite basis for this ideal is computed using techniques from Gr\"obner bases and elimination theory.
This step is similar to the case of the invariant ideal for solvable loops, see~e.g.,~\cite{Rodriguez04,Kovacs08}.
\end{itemize}
In conclusion, we infer a \emph{finite representation of the ideal of  polynomial invariants for loops with unsolvable recurrence operators}.

%% file: 06loopsynthesis.tex
\section{Synthesising Solvable Loops from Unsolvable Loops}
\label{sec:loopsynthesis}

In previous sections, we introduced a new technique to compute invariants for unsolvable loops; that is, loops containing defective variables.
An orthogonal challenge is to synthesise a solvable loop from an unsolvable loop that preserves or over-approximates given specifications.

In this section we establish, with Algorithm~\ref{alg:loop-synthesis}, a new method to synthesise a solvable loop $\P'$ from an unsolvable loop $\P$.
The solvable loop $\P'$ over-approximates the behaviour of $\P$ in the sense that every polynomial invariant of $\P'$ is an invariant of $\P$.
Moreover, we show the invariants among effective variables of $\P$ and $\P'$ coincide.
The following example illustrates the main idea leading to Algorithm~\ref{alg:loop-synthesis}.

\medskip

\begin{example}\label{ex:squares_synthesized}
Example~\ref{ex:squares_solved} showed how to synthesise the polynomial $S = x + y$ of defective variables $x$ and $y$ for the loop in Figure~\ref{fig:running:squares} such that $S$ admits a closed-form.
In this case, the polynomial of program variables $S$ satisfies the linear inhomogeneous recurrence $S(n+1) = 2S(n) - 3z(n) + 3$.
We can use this recurrence to construct a solvable loop from the unsolvable from Figure~\ref{fig:running:squares} that captures the dynamics of the only effective variable $z$ as well as $S$:

\medskip

\noindent
\begin{minipage}{\textwidth}
\begin{minipage}{0.3\textwidth}
\begin{algorithmic}
    \STATE \(z \leftarrow 0\)
    \WHILE {$\star$}
        \STATE \(z \leftarrow 1 - z\)
        \STATE \(x \leftarrow 2x + y^2 + z\)
        \STATE \(y \leftarrow 2y - y^2 + 2z\)
    \ENDWHILE
\end{algorithmic}
\end{minipage}
\begin{minipage}{0.2\textwidth}
\centering
{\footnotesize Algorithm~\ref{alg:loop-synthesis}} \\
$\longmapsto$
\end{minipage}
\begin{minipage}{0.4\textwidth}
\begin{algorithmic}
    \STATE \(z \leftarrow 0\)
    \STATE \(\texttt{s} \leftarrow x_0 + y_0\)
    \WHILE {$\star$}
        \STATE \(
            \begin{pmatrix}
            z \\ \texttt{s}
            \end{pmatrix}
            \leftarrow
            \begin{pmatrix}
            1 - z \\ 2\texttt{s} - 3z + 3
            \end{pmatrix}
        \)
    \ENDWHILE
\end{algorithmic}
\end{minipage}
\end{minipage}

\medskip

Our algorithmic approach, as given in Algorithm~\ref{alg:loop-synthesis}, synthesises the solvable loop from the unsolvable loop on the left. Such a synthesis step is implemented within our tool (Section~\ref{section:experiments}), allowing us to synthesize the above solvable loop in about $1$ second.
\end{example}

\begin{algorithm}[H]
\caption{Solvable Loop Synthesis}
    \begin{algorithmic}[1]
        \REQUIRE Unsolvable loop $\P$ with recurrence operator $\cR$, degree $d \in \N$
        \ENSURE Solvable loop $\P'$
        \STATE Compute $E(\P)$ and $D(\P)$ using Algorithm~\ref{alg:nmcv-alg}.
        \STATE Fix candidate polynomial $\s{n}$ of degree $d$ (as in Section~\ref{sec:invariants} \eqref{eqn:sn-eq}).
        \STATE Solve for coefficients in \\ \quad $\s{n{+}1} = \kappa\s{n} + \sum_{M \in \cE_{n}(\P)} c_M M$ (as in Section~\ref{sec:invariants} \eqref{eqn:main-rec}).
        \STATE $\texttt{initial} = []$
        \STATE $\texttt{body\_left} = []$
        \STATE $\texttt{body\_right} = []$
        \STATE \COMMENT{Add all effective variables to new loop}
        \FOR{$x \in E(\P)$}
            \STATE $\texttt{initial.append}(x \leftarrow x_0)$
            \STATE $\texttt{body\_left.append}(x)$
            \STATE $\texttt{body\_right.append}(\cR[x])$
        \ENDFOR
        \STATE \COMMENT{Add well-behaved combination of defective monomials}
        \IF{$S \neq 0$}
            \STATE Choose a fresh symbol $\texttt{s}$.\label{alg:loop-synthesis:fresh}
            \STATE $\texttt{initial.append}(\texttt{s} \leftarrow \s{0})$
            \STATE $\texttt{body\_left.append}(\texttt{s})$
            \STATE $\texttt{body\_right.append}(\kappa \texttt{s} + \sum_{M \in \cE(\P)} c_M M)$\label{alg:loop-synthesis:s}
        \ENDIF
        \STATE $\P' \leftarrow$ $\texttt{initial}$ \lq\lq while $\star$ do\rq\rq\ $\texttt{body\_left} \leftarrow \texttt{body\_right}$ \lq\lq end while\rq\rq
        \RETURN $\P'$
    \end{algorithmic}
\label{alg:loop-synthesis}
\end{algorithm}

The inputs for Algorithm~\ref{alg:loop-synthesis} are an unsolvable loop $\P$ with recurrence operator $\cR$ and a fixed degree $d \in \N$; the algorithm outputs a solvable loop $\P'$, if it exist.
Briefly, the algorithm invokes the invariant synthesis procedure from Section~\ref{alg:loop-synthesis} (for degree $d$) and constructs the loop $\P'$ from $\P$ by removing all the defective variables and then introducing a new variable $\texttt{s}$ that models an invariant among defective monomials, if such an invariant exists.
The recurrence operator $\cR'$ associated with the synthesised loop $\P'$ is the canonical recurrence operator: $\cR'$ maps every program variable to its assignment.

\medskip

\begin{lemma}[Soundness]\label{lemma:synthesis:solvable}
The loop $\P'$ returned by Algorithm~\ref{alg:loop-synthesis} is solvable.
Moreover, we have $E(\P) \subseteq \vars{\P'}$.
\end{lemma}
\begin{proof}
Using our characterisation of solvable and unsolvable loops in terms of effective and defective variables (Theorem~\ref{thm:def-char}), we show that all variables in $\P'$ are effective.
The variables of $\P'$ consist of the effective variables of $\P$ as well as the fresh variable $\texttt{s}$: $\vars{\P'} = E(\P) \cup \{ \texttt{s} \}$, or  $\vars{\P'} = E(\P)$ if no invariant among defective monomials with degree $d$ exists.

First, every $x \in E(\P)$ necessarily remains effective for the synthesised program $\P'$:
in the dependency graph of $\P'$, the variable $x$ cannot occur in a cycle containing $\texttt{s}$, because $\texttt{s}$ is a fresh variable and hence cannot appear in the recurrence $\cR[y]$ for any $y \in E(\P)$ (Algorithm~\ref{alg:loop-synthesis} line \ref{alg:loop-synthesis:fresh}).
Hence, if \(z \in E(\P)\cap D(\P')\), then there must exist a cycle in the dependency graph of $\P'$ with a non-linear edge \((x,y)\) (i.e., \(\mathcal{L}(x,y)=N\)) such that every vertex in the cycle is in $E(\P)$.
Consequently, this cycle is also present in the dependency graph of the original program $\P$.
This means that not all variables in $E(\P)$ are effective, which is a contradiction.

Second, the variable $\texttt{s}$ is effective.
Because $\texttt{s}$ is a fresh variable, the only incoming edge of $\texttt{s}$ in the dependency graph of $\P'$ is a linear self-loop (Algorithm~\ref{alg:loop-synthesis} line \ref{alg:loop-synthesis:s}).
Hence $\texttt{s}$ cannot occur in a cycle with a non-linear edge.
Moreover, all outgoing edges point to effective variables.
Thus $\texttt{s}\in E(\P')$ is effective.
\end{proof}

With the next two lemmas, we prove that the  loop $\P'$ synthesised by Algorithm~\ref{alg:loop-synthesis} over-approximates the invariants of the unsolvable loop $\P$.
Let $\Inv{\P}$ and $\Inv{\P'}$ denote the invariant ideals of $\P$ and $\P'$ respectively.
The next lemma establishes that the synthesised loop $\P'$ is \emph{complete with respect to invariants among effective variables}.

\medskip

\begin{lemma}[Completeness with respect to Effective Variables]\label{lemma:synthesis:effective}
The invariant ideals of $\P$ and $\P'$ coincide when restricted on the effective variables of $\P$.
That is,
\begin{equation*}
    \Inv{\P} \cap \overline{\Q}[E(\P)] =  \Inv{\P'} \cap \overline{\Q}[E(\P)].
\end{equation*}
\end{lemma}
\begin{proof}
By the construction of $\P'$, every $x \in E(\P)$ is also a variable of $\P'$.
To distinguish the variable $x$ in $\P$ from the variable $x$ in $\P'$ we refer to the latter by $x'$.
We will show that for every $x \in E(\P)$ the sequences $\seq{x(n)}{n}$ and $\seq{x'(n)}{n}$ coincide.
If this is the case, the polynomial invariants for the programs $\P$ and $\P'$ among the variables $E(\P)$ necessarily coincide as well.

Let $\cR$ and $\cR'$ be the recurrence operators associated to $\P$ and $\P'$, respectively.
For every $x \in E(\P)$ we have $\cR[x] = \cR'[x']$ and $x_0 = x'_0$ by the construction of~$\P'$.
Furthermore, by Corollary~\ref{cor:effective}, defective variables cannot occur in $\cR[x]$ for any $x \in E(\P)$.
Moreover, the fresh variable $\texttt{s}$ cannot occur in $\cR'[x']$ by the construction of~$\P'$.
Hence the two systems of first-order recurrences $\{ x(n{+}1) = \cR[x] \mid x \in E(\P) \}$ and $\{ x'(n{+}1) = \cR'[x'] \mid x \in E(\P) \}$ together with the initial values $\{ x_0 \mid E(\P) \}$ and $\{ x'_0 \mid E(\P) \}$ induce the same sequences $\seq{x(n)}{n}$ and $\seq{x'(n)}{n}$ for all $x \in E(\P)$.
\end{proof}

We note that when  an invariant among defective monomials of degree \(d\) does not exist, the variables of the synthesised loop $\P'$ are precisely the effective variables of $\P$.
In this case,  Lemma~\ref{lemma:synthesis:effective} fully characterises the relationship between the invariants of $\P$ and $\P'$. 

In the following, let us consider the complementary case, namely when an invariant among the  defective monomials of $\P$ exists.
Hence,  the synthesised loop $\P'$ contains the additional fresh variable $\texttt{s}$ modelling the behaviour of this invariant. 
With the next lemma, we confirm that the synthesised loop $\P'$ is indeed a \emph{sound over-approximation of the unsolvable loop $\P$}.
We show that every invariant of $\P'$ is also an invariant of $\P$.
The program variable $\texttt{s} \in \vars{\P'}$ introduced by Algorithm~\ref{alg:loop-synthesis} is however not a program variable of $\P$; nevertheless, $\texttt{s}$  models the polynomial $S$ of defective variables in $\P$.
Hence, to compare the invariant ideals of $\P$ and $\P'$, we need to \lq\lq substitute\rq\rq\ $\texttt{s}$ by the polynomial of defective variables it models.
This can be done by adding the equation $\texttt{s} = S$ ($\texttt{s} - S = 0$) to the invariant ideal of $\P'$ and restricting the resulting ideal to \vars\P:

\begin{equation}\label{eq:synthesized-ideal}
I( \Inv{\P'} \cup \{ \texttt{s} - S \} ) \cap \overline{\Q}[\vars\P].
\end{equation}

\medskip

\begin{lemma}[Over-Approximation]\label{lemma:synthesis:approx}
Let $J$ be the ideal in \eqref{eq:synthesized-ideal} constructed from the invariant ideal of $\P'$ by replacing the program variable $\texttt{s}$ by the polynomial of defective variables it models. Then, $J \subseteq \Inv{\P}$.
\end{lemma}
\begin{proof}

As argued in the proof of Lemma~\ref{lemma:synthesis:effective}, for every $x \in E(\P)$, the sequences corresponding to the variable $x$ in both $\P$ and $\P'$ coincide; that is, $\seq{x(n)}{n} \equiv \seq{x'(n)}{n}$.
Furthermore, we have $\vars{\P'} = E(\P) \cup \{ \texttt{s} \}$.
The fresh variable $\texttt{s}$ in $\P'$ models the polynomial $S \in \overline{\Q}[\vars\P]$.
Let $\seq{\texttt{s}(n)}{n}$ be the sequence induced by the program variable $\texttt{s}$ in $\P'$ and, likewise, $\seq{S(n)}{n}$ the sequence induced by the polynomial $S \in \overline{\Q}[\vars\P]$.
Then, by the construction of $\P'$, we have $\seq{\texttt{s}(n)}{n} \equiv \seq{S(n)}{n}$.

Let $Q \in \Inv{P'}$.
\re{By the definition of the ideal $J$ in \eqref{eq:synthesized-ideal},} it holds that 
\begin{equation*}
    Q \in \Inv{P'} \iff Q\{\texttt{s} \mapsto S\}\in J
\end{equation*} (where the notation indicates that $S$ is substituted for $\texttt{s}$).

Now, $Q$ is a polynomial relation among the sequences induced by the variables in $\vars{\P'}$.
Because $\seq{x(n)}{n} \equiv \seq{x'(n)}{n}$ for every $x \in E(\P)$ and $\seq{\texttt{s}(n)}{n} \equiv \seq{S(n)}{n}$, the polynomial $Q\{\texttt{s} \mapsto S\}$ represents a relation among the sequences induced by the variables in $\vars\P$.
Hence we have $Q \in \Inv{\P}$.
\end{proof}

\begin{remark}
Our loop synthesis procedure given in Algorithm~\ref{alg:loop-synthesis} computes a single invariant among defective monomials (if such an invariant exists) of an unsolvable loop~$\P$.
The results in this section naturally generalise to multiple invariants among defective monomials, as follows: 
for every invariant $I$, add a fresh variable $\texttt{s}_I$ modelling the behaviour of $I$ to the synthesised program $\P'$.
Note that with each additional invariant \(I\) added, the dynamics of the synthesised loop \(\P'\) more closely resembles that of the unsolvable loop~\(\P\).
\end{remark}

%% file: 07probabilistic.tex

\section{Adjustments for Unsolvable Operators in Probabilistic Programs} 
\label{sec:probabilistic}

\subsection{Defective Variables in Probabilistic Loop Models}
The works~\cite{Moosbruggeretal2022,BKS19} defined recurrence operators for probabilistic loops.
Specifically, a recurrence operator is defined for loops with polynomial assignments, probabilistic choice, and drawing from common probability distributions with constant parameters.
Recurrences for deterministic loops model the precise values of program variables.
For probabilistic loops, this approach is not viable, due to the stochastic nature of the program variables.
Thus a recurrence operator for a probabilistic loop models \emph{(higher) moments} of program variables.
As illustrated in Example~\ref{ex:NLMarkov-1}, the recurrences of a probabilistic loop are taken over expected values of program variable monomials.

The authors of \cite{Moosbruggeretal2022,BKS19} explicitly excluded the case of circular non-linear dependencies to guarantee computability.
However, in contrast to our notions in Sections~\ref{section:recurrences}, they defined variable dependence not on the level of recurrences but on the level of assignments in the loop body.
To use the notions of effective and defective variables for probabilistic loops, we follow the same approach and base the dependency graph on assignments rather then recurrences.
We illustrate the necessity of this adaptation in the following example.

\medskip 

\begin{example}\label{ex:prob-phenomena}
\re{
A probabilistic assignment $x \leftarrow a\ \{p\}\ b$ intuitively means that $x$ is assigned $a$ with probability $p$ and $b$ with probability $1{-}p$.
}
Consider the following probabilistic loop and associated set of first-order recurrence relations in terms of the expectation operator \(\E(\wc)\).

\medskip
\noindent
\begin{minipage}{\linewidth}
\begin{minipage}[b]{0.32\textwidth}
\begin{algorithmic}
\WHILE{$\star$}
\STATE \(y \leftarrow 4y(1-y)\)
\STATE \(x \leftarrow x-y\ \{1/2\}\ x+y\)
\ENDWHILE
\end{algorithmic}
\bigskip
\end{minipage}
\hfill\vline\hfill
\begin{minipage}[b]{0.62\textwidth}
\begin{align*}
    \E(y_{n+1}) &= 4\E(y_n) - 4\E(y_n^2) \\
    \E(x_{n+1}) &= \E(x_{n}) \\
    \E(x^2_{n+1}) &= \re{\frac{1}{2}\E((x_n - y_{n+1})^2) + \frac{1}{2}\E((x_n + y_{n+1})^2)} \\
    &= \E(x^2_{n}) + \E(y^2_{n+1})
\end{align*}
\end{minipage}
\end{minipage}
\bigskip

It is straightforward to see that variable \(y\) is defective from the deterministic update 
\(y \leftarrow 4y(1-y)\) with its characteristic non-linear self-dependence.
Moreover, $y$ appears in the probabilistic assignment of $x$:
However, due to the particular form of the assignment, the recurrence of $\E(x_n)$ does not contain $y$.
Nevertheless, $y$ appears in the recurrence of $\E(x^2_n)$.
This phenomenon is specific to the probabilistic setting.
For deterministic loops, it is always the case that if the values of a program variable $w$ do not depend on defective variables, then neither do the values of any power of $w$.
\end{example}

In light of the phenomenon exhibited in Example~\ref{ex:prob-phenomena}, we adapt our notion of \emph{variable dependency}, for probabilistic loops.
Without loss of generality, we assume that every program variable has exactly one assignment in the loop body.
Let $\P$ be a probabilistic loop and $x, y \in \vars\P$.
We say $x$ \emph{depends on} $y$, if $y$ appears in the assignment of $x$.
Additionally, the dependency is \emph{linear} if all occurrences of $y$ in the assignment of $x$ are linear, else the dependency is \emph{non-linear}.
Further, we consider the transitive closure of variable dependency analogous to deterministic loops and Definition~\ref{def:dependency}. 

With variable dependency thus defined, the dependency graph and the notions of effective and defective variables follow immediately.
Analogous to our characterisation of unsolvable recurrence operators in terms of defective variables for deterministic loops, \emph{all (higher) moments} of effective variables of probabilistic loops can be described by a system of linear recurrences~\cite{Moosbruggeretal2022,BKS19}.
For defective variables this property will generally fail
For instance, in Example~\ref{ex:prob-phenomena}, the variable $x$ is now classified as defective and $\E(x^2_n)$ cannot be modelled by linear recurrences for some initial values.

The only necessary change to the invariant synthesis algorithm from Section~\ref{sec:invariants} is as follows: instead of program variable monomials, we consider expected values of program variable monomials.
Now, our invariant synthesis technique from Section~\ref{sec:invariants} can also be applied to probabilistic loops to synthesise combinations of expected values of defective monomials that do satisfy a linear recurrence.

\subsection{Synthesising Solvable Probabilistic Loops} 

In Algorithm~\ref{alg:loop-synthesis} we introduced a procedure, utilising our new invariant synthesis technique from Section~\ref{sec:invariants}, to over-approximate an unsolvable loop by a solvable loop.
The inputs to Algorithm~\ref{alg:loop-synthesis} are an unsolvable loop with a recurrence operator and a natural number specifying a fixed degree.
As mentioned, our invariant synthesis procedure is also applicable to probabilistic loops using the recurrence operator modelling moments of program variables introduced in~\cite{Moosbruggeretal2022,BKS19}.
Hence, Algorithm~\ref{alg:loop-synthesis} can also be used to synthesise solvable loops from unsolvable probabilistic loops.
In the probabilistic case, however, the invariants computed by our approach are over \emph{moments of program variables}.
Therefore, the invariant ideal of probabilistic loops describes polynomial relations among a given set of moments of program variables, such as the expected values.
Consequently, the loop synthesised by Algorithm~\ref{alg:loop-synthesis} for a given probabilistic loop will be deterministic and model the dynamics of moments of program variables of the probabilistic loop.

\medskip 

\begin{example}\label{ex:prob-loop-synthesis}
Recall the program \(\P_\textrm{MC}\) of  Example~\ref{ex:NLMarkov-1}.
An invariant synthesised by our approach in Section~\ref{sec:invariants} with degree $1$ is \(\E(x_n - y_n) = \frac{5^n}{6^n}(x_0 - y_0)\).
Hence the solvable loop synthesised by Algorithm~\ref{alg:loop-synthesis} for \(\P_\textrm{MC}\) with input degree $1$ is
\begin{algorithmic}
\STATE \(\texttt{f} \leftarrow x_0 + y_0\)
\WHILE{$\star$}
\STATE \(s \leftarrow \frac{1}{2}\) \vspace{0.25em}
\STATE \(\texttt{f} \leftarrow \frac{5}{6} \texttt{f} \)
\ENDWHILE
\end{algorithmic}
where $\texttt{f}$ is the fresh variable introduced by Algorithm~\ref{alg:loop-synthesis} modelling \(\E(x_n - y_n)\).
Our approach from Algorithm~\ref{alg:loop-synthesis} synthesises this solvable loop from \(\P_\textrm{MC}\), using less than $0.5$ second within our implementation  (Section~\ref{section:experiments}).
\end{example}

%% file: 08applications.tex
\section{Applications of Unsolvable Operators towards Invariant Generation}%
\label{section:case-studies}
Our approach automatically generates invariants for programs with defective variables (Section~\ref{sec:invariants}), and pushes the boundaries of both theory and practice of invariant generation:  we introduce and incorporate defective variable analysis  into the state-of-the-art methodology for reasoning about solvable loops, complementing thus existing methods, see e.g.,~\cite{Rodriguez04,Kovacs08,Kincaid18,Humenberger18},  in the area.   
As such, the class of unsolvable loops that can be handled by our work extends (aforementioned) existing approaches on polynomial invariant generation. 
The experimental results of our approach (see Section~\ref{section:experiments}) demonstrate the efficiency and scalability of our work in  deriving invariants for unsolvable loops.  
Since our approach to loops via recurrences is generic, 
we can deal with emerging applications of programming paradigms such as: transitions systems and statistical moments in probabilistic programs; and reasoning about biological systems. 
We showcase these applications in this section and also exemplify the limitations of our work.
In the sequel,  we write $\E(t)$ to refer to the expected value of an expression $t$, and denote by $\E(t_n)$ (or \(\E(t(n))\)) the expected value of \(t\) at loop iteration $n$.

\medskip

\begin{example}[Moments of Probabilistic Programs~\cite{schreuder2019invariants}\label{ex:2:revised}]
\hl{As mentioned in Example~\ref{ex:prob-loop-synthesis}, \(\E(x_n - y_n) = \frac{5^n}{6^n}(x_0 - y_0)\) is an invariant for the program \(\P_\textrm{MC}\) introduced in Example~\ref{ex:NLMarkov-1}.}
Closed-form solutions for higher order expressions are also available; for example,
\begin{equation*}
\E((x_n-y_n)^d) = \frac{(2^d + 3^d)^n}{2^n\cdot 3^{dn}} (x_0 - y_0)^d
\end{equation*}
refers to the $d$th moment of $x(n)-y(n)$.
While the work in~\cite{schreuder2019invariants} uses martingale theory to synthesise the above  invariant (of degree 1), 
our approach automatically generates such invariants  over higher-order moments (see Table~\ref{table:experiments2}).
We note to this end that the defective variables in \(\P_\textrm{MC}\) are precisely $x$ and $y$ as can be seen from their mutual non-linear interdependence. Namely,  we have \(D(\P_\textrm{MC}) = \{x, y\}\) and $E(\P_\textrm{MC}) = \{s\}$.
\end{example}

\medskip

\begin{example}[\texttt{non-lin-markov-2}] \label{ex:non-lin-markov-2}
We give a second example of a non-linear Markov chain.
We analyse the moments of this probabilistic program in the next section.
\begin{algorithmic}
\STATE \(x,y \leftarrow 0,1\)
\WHILE{$\star$}
\STATE \(s \leftarrow \bernoulli{1/2}\)
\IF{\(s=0\)}
\STATE \(   
\begin{pmatrix}
x\\
y
\end{pmatrix}
\leftarrow
\begin{pmatrix}
\frac{4}{10}(x + xy)\\
\frac{4}{10}({1}{3}x + \frac{2}{3}y + xy)
\end{pmatrix}
\)
\ELSE
\STATE \(   
\begin{pmatrix}
x\\
y
\end{pmatrix}
\leftarrow
\begin{pmatrix}
\frac{4}{10}(x + y +  \frac{2}{3}xy) \\
\frac{4}{10}(2y + \frac{2}{3}xy)
\end{pmatrix}
\)
\ENDIF
\ENDWHILE
\end{algorithmic}
\end{example} 
    
\smallskip
       
\begin{example}[Biological Systems~\cite{britton2002deciding}]\label{ex:bees}%
A model for the decision-making process of swarming bees choosing one nest-site from a selection of two is introduced in~\cite{britton2002deciding} and further studied in~\cite{dreossi2016parallelotope,sankaranarayanan2020reasoning}.  
Previous works have computed probability distributions for this model~\cite{sankaranarayanan2020reasoning}. 
 The (unsolvable) loop is a discrete-time model with  five classes of bees (each represented by a program variable).
The coefficient \(\Delta\) is the length of the time-step in the model and the remaining coefficients parameterise the rates of change.
All coefficients here are symbolic.
{
    \begin{algorithmic}
        \STATE \(
        \begin{pmatrix}
        x\\
        y_1\\
        y_2\\
        z_1\\
        z_2\\
        \end{pmatrix}
        \leftarrow
        \begin{pmatrix}
        \normal{475,5}\\
        \uniform{350,400}\\
        \uniform{100,150}\\
        \normal{35, 1.5}\\
        \normal{35, 1.5}\\
        \end{pmatrix}
        \)
        \WHILE{$\star$}
        \STATE \(
        \begin{pmatrix}
        x\\
        y_1\\
        y_2\\
        z_1\\
        z_2\\
        \end{pmatrix}
        \leftarrow
        \begin{pmatrix}
       x - \Delta(\beta_1 x y_1 +  \beta_2 x y_2) \\
       y_1 + \Delta(\beta_1 x y_1 - \gamma y_1 + \delta\beta_1 y_1 z_1 + \alpha\beta_1 y_1 z_2) \\
       y_2 + \Delta(\beta_2 x y_2 - \gamma y_2 + \delta\beta_2 y_2 z_2 + \alpha\beta_2 y_2 z_1)\\
       z_1 + \Delta(\gamma y_1 - \delta\beta_1 y_1 z_1 - \alpha\beta_2 y_2 z_1)\\
       z_2  + \Delta(\gamma y_2 - \delta\beta_2 y_2 z_2 - \alpha\beta_1 y_1 z_2)
        \end{pmatrix}
        \) 
        \ENDWHILE
    \end{algorithmic}
   } 
 We note that the  model in~\cite{sankaranarayanan2020reasoning} uses truncated Normal distributions,  as~\cite{sankaranarayanan2020reasoning} is limited to finite supports for the program variables, which is not the case with our work. 
 
 In the loop above, each of the variables exhibits non-linear self-dependence, and so the variables are partitioned into  \(D(\P) = \{x, y_1, y_2, z_1, x_2 \} \) and \(E(\P) = \emptyset \). While the recurrence operator of the loop above is unsolvable, our approach infers  polynomial loop invariants using defective variable reasoning (Section~\ref{sec:invariants}). 
Namely, we generate the following closed-form solutions over expected values of program variables:
    \begin{align*}
        \E(x(n) + y_1(n) + y_2(n) + z_1(n) + z_2(n)) &= 1045, \\
        \E((x(n) + y_1(n) + y_2(n) + z_1(n) + z_2(n))^2) &= {3277349}/{3}, \quad \text{and} \\
        \E((x(n) + y_1(n) + y_2(n) + z_1(n) + z_2(n))^3) &= 1142497455.
    \end{align*}
    One can interpret such invariants in terms of the biological assumptions in the model.
    Take, for example, the fact that \(\E(x(n) + y_1(n) + y_2(n) + z_1(n) + z_2(n))\) is constant.
    This invariant is in line with the assumption in the model that the total population of the swarm is constant.
    In fact, our invariants reflect the behaviour of the system in the original \emph{continuous-time} model proposed in \cite{britton2002deciding}, because our approach is able to process all coefficients (most importantly $\Delta$) as symbolic constants.
 \end{example}

\medskip 
\begin{example}[Probabilistic Transition Systems~\cite{schreuder2019invariants}]\label{ex:pts}
    Consider the following probabilistic loop  modelling a \emph{probabilistic transition system} from  \cite{schreuder2019invariants}:

    \begin{algorithmic}
    \WHILE{$\star$}
    \STATE \(
        \begin{pmatrix}
            a \\
            b
        \end{pmatrix} \leftarrow
        \begin{pmatrix}
            \normal{0, 1} \\
            \normal{0, 1}
        \end{pmatrix}
        \)
    \STATE \(
        \begin{pmatrix}
            x \\
            y
        \end{pmatrix} \leftarrow
        \begin{pmatrix}
            x + axy \\
            y + bxy
        \end{pmatrix}
        \)  
    \ENDWHILE
    \end{algorithmic}
    While~\cite{schreuder2019invariants} uses martingale theory to synthesise a degree one invariant of the form \(a \E(x_k) + b \E(y_k) = a \E(x_0) + b \E(y_0)\), our work automatically generates  invariants over higher-order moments involving the defective variables \(x\) and \(y\), as presented in Table~\ref{table:experiments2}. 
\end{example}

\medskip

The next example demonstrates an unsolvable loop whose recurrence operator cannot (yet) be handled by our work. 

\medskip

\begin{example}[Trigonometric Updates]
    As our approach is limited to polynomial updates of the program variables, the loop below cannot be handled by our work: 

    \begin{algorithmic}
    \WHILE{$\star$}
    \STATE \(
    \begin{pmatrix}
        x \\
        y
    \end{pmatrix}
    \leftarrow
    \begin{pmatrix}
        \cos(x) \\
        \sin(x)
    \end{pmatrix}
    \)
    \ENDWHILE
    \end{algorithmic}
    Note the trigonometric functions are transcendental, from which it follows that one cannot generally obtain closed-form solutions for the program variables.
    Nevertheless, this program does admit polynomial invariants in the program variables; for example,  \(x^2 + y^2 = 1\). 
     Although our definition of a defective variables does not apply here, we could say the variable $x$ here is \emph{somehow defective}: while the exact value of \(\sin(x)\) cannot be computed, it could be approximated using  power series. Extending our work with more general notions of defective variables is an interesting line for future work.
\end{example}

\medskip

\hl{Examples~\ref{ex:squares+}--\ref{ex:deg-d} (below) are custom-made benchmarks.}
We have tailored these benchmarks to demonstrate the flexibility and applicability of our method to the current state of the art. 
Our experimental analysis is delayed to Section~\ref{section:experiments}.

\medskip
\begin{example}[\texttt{squares+}] \label{ex:squares+}
\hfill
\begin{algorithmic}
\STATE \(s, x, y, z \leftarrow 0, 2, 1, 0\)
\WHILE{$\star$}
\STATE \(s \leftarrow \bernoulli{1/2}\)
\STATE \(z \leftarrow z-1\ \{1/2\}\ z + 2\)
\STATE \(x \leftarrow 2x + y^2 + s + z\)
\STATE \(y \leftarrow 2y - y^2 +2s\)
\ENDWHILE
\end{algorithmic}
\end{example}

\medskip

\begin{example}[\texttt{prob-squares}] \label{ex:prob-squares}
\hfill
\begin{algorithmic}
\STATE \(g \leftarrow 1\)
\WHILE{$\star$}
\STATE \(g \leftarrow \uniform{g,  2g}\)
\STATE \(
\begin{pmatrix}
a \\
b \\
c
\end{pmatrix} \leftarrow
\begin{pmatrix}
a^2 + 2bc -df +b \\
df -a^2 +2bd +2c \\
g -bc -bd +\frac{1}{2}a
\end{pmatrix}
\)
\ENDWHILE
\end{algorithmic}
\end{example}

\medskip
\begin{example}[\texttt{squares-squared}] \label{ex:squares-squared}
\hfill
\begin{algorithmic}
\WHILE{$\star$}
\STATE \(
\begin{pmatrix}
x \\
y \\
z \\
m
\end{pmatrix} \leftarrow
\begin{pmatrix}
xyz + x^2 \\
2y + z - x^2 + 3ym z^2 \\
\frac{3}{2}x + \frac{3}{2}z + \frac{1}{2}y + \frac{1}{2}x^2 \\
\frac{2}{3}z +3m - \frac{1}{3}x^2 -\frac{1}{3}xyz -ymz^2
\end{pmatrix}
\)
\ENDWHILE
\end{algorithmic}
\end{example}

\medskip


\begin{example}[\texttt{deg-d}] \label{ex:deg-d}
The benchmarks \texttt{deg-5}, \texttt{deg-6}, \texttt{deg-7}, \texttt{deg-8}, \texttt{deg-9}, and \texttt{deg-500} are parameterised by the degree \(d\) in the following program.

\begin{algorithmic}
\STATE \(x, y \leftarrow 1, 1\)
\WHILE{$\star$}
\STATE \(z \leftarrow \normal{0, 1}\)
\STATE \(
    \begin{pmatrix}
        x \\
        y
    \end{pmatrix} \leftarrow
    \begin{pmatrix}
        2x^d + z + z^2 \\
        3x^d + z + z^2 + z^3
    \end{pmatrix}
    \)
\ENDWHILE
\end{algorithmic}
\end{example}

\medskip

The following set of examples are taken from the literature on the theory of trace maps.
Arguably the most famous example is the classical \emph{Fibonacci Trace Map} \(f\colon \R^3\to\R^3\) given by \(f(x,y,z)=(y,z,2yz-x)\) (\cref{ex:fib1} below); said map has garnered the attention of researchers in fields as diverse as invariant analysis, representation theory, geometry, and mathematical physics (cf.\ the survey papers \cite{baake1993trace,roberts1994trace}).
From a computational viewpoint, trace maps arise from substitution rules on matrices (see, again, the aforementioned survey papers).
Given two matrices \(A,B\in\SL(2,\R)\) (the group of \(2\times 2\) matrices with unit determinant), consider the following substitution rule on strings of matrices: \(A\mapsto B\) and \(B\mapsto AB\).
The classical Fibonacci Trace Map is determined by the action of this substitution on the traces of the matrices; i.e.,
    \begin{equation*}
        f\bigl(\tr(A),\tr(B), \tr(AB)\bigr) = \bigl(\tr(B), \tr(AB), -\tr(A) + 2\tr(B)\tr(AB)  \bigr).
    \end{equation*}
Further examples of trace maps (\cref{ex:fib2} and \cref{ex:fib3} below) are constructed from similar substitution rules on strings of matrices.

\re{
Let 
    \begin{align*}
        I_1(x,y,z) &= x^2 + y^2 + z^2 - 2xyz, \\
        I_2(x,y,z) &= 4x^2y - 2xz - y, \\
        I_3(x,y,z) &= x^2 +y^2 +z^2 - xyz - xy - xz - yz -x - y - z.
    \end{align*}
Then our work generates the cubic polynomial invariants \(I_1(x,y,z) - I_1(x_0, y_0, z_0)=0\), \(I_2(x,y,z)-I_2(x_0,y_0,z_0)=0\), and \(I_3(x,y,z)-I_3(x_0,y_0,z_0)=0\) for Example~\ref{ex:fib1}, Example~\ref{ex:fib2}, and Example~\ref{ex:fib3}, respectively.
}

\medskip

\begin{example}[\texttt{fib1}] \label{ex:fib1} \hfill
\begin{algorithmic}
\WHILE{$\star$}
\STATE \(
    \begin{pmatrix}
        x \\
        y \\
        z
    \end{pmatrix}
    \leftarrow
    \begin{pmatrix}
        y \\
        z \\
        2yz -x
    \end{pmatrix}
    \)
\ENDWHILE
\end{algorithmic}
\end{example}

\medskip

\begin{example}[\texttt{fib2}] \label{ex:fib2}
A Generalised Fibonacci Trace Map

\begin{algorithmic}
\WHILE{$\star$}
\STATE \(
    \begin{pmatrix}
        x \\
        y \\
        z
    \end{pmatrix}
    \leftarrow
    \begin{pmatrix}
        y \\
        2xz - y \\
        4xyz - 2x^2 -2y^2 + 1
    \end{pmatrix}
    \)
\ENDWHILE
\end{algorithmic}

\medskip

\end{example}

\begin{example}[\texttt{fib3}] \label{ex:fib3}
A second Generalised Fibonacci Trace Map

\begin{algorithmic}
\WHILE{$\star$}
\STATE \(
    \begin{pmatrix}
        x \\
        y \\
        z
    \end{pmatrix}
    \leftarrow
    \begin{pmatrix}
        1 + x + y + xy -z \\
        x \\
        y
    \end{pmatrix}
    \)
\ENDWHILE
\end{algorithmic}
\end{example}

\medskip

Examples~\ref{ex:markov-triples-toggle}--\ref{ex:markov-triples-random} are while loops that generate \emph{Markov triples} \cite[Chapter II.3]{cassels1972markov}; that is,
at every iteration each loop variable takes an integer value that appears in a Diophantine solution of the Markov equation \(x^2+y^2+z^2=3xyz\).

\medskip

\begin{example}[\texttt{markov-triples-toggle}] \label{ex:markov-triples-toggle}
A while loop that generates an infinite sequence of nodes on the Markov tree (the walk alternates between `upper' and `lower' branches).
\begin{algorithmic}
\STATE x,y,z = 1,1,2;
\STATE \(\textrm{branch} = 0\);
\WHILE{$\star$}
\IF{\(\textrm{branch} = 0\)}
\STATE \(
    \begin{pmatrix}
        x \\
        y \\
        z
    \end{pmatrix}
    \leftarrow
    \begin{pmatrix}
        x \\
        3xy -z \\
        y
    \end{pmatrix};\)
\STATE \(\textrm{branch}=1\);
\ELSE   
\STATE \(
    \begin{pmatrix}
        x \\
        y \\
        z
    \end{pmatrix}
    \leftarrow
    \begin{pmatrix}
        y \\
        3yz -x \\
        z
    \end{pmatrix};\)
\STATE \(\textrm{branch}=0\);
\ENDIF
\ENDWHILE
\end{algorithmic}
For this example, our approach %
generates the polynomial invariant, the \emph{Markov equation}, given by \(x^2 + y^2 + z^2 - 3xyz=0\).
\end{example}

\medskip

\begin{example}[\texttt{markov-triples-random}] \label{ex:markov-triples-random}
A while loop that simulates a Bernoulli walk on the Markov tree.
\begin{algorithmic}
\STATE x,y,z = 1,1,2
\WHILE{$\star$}
\STATE \(p \leftarrow \bernoulli{1/2}\)
\IF{\(p=1\)}
\STATE \(
    \begin{pmatrix}
        x \\
        y \\
        z
    \end{pmatrix}
    \leftarrow
    \begin{pmatrix}
        x \\
        3xy-z \\
        y
    \end{pmatrix}
    \)
\ELSE
\STATE \(
    \begin{pmatrix}
        x \\
        y \\
        z
    \end{pmatrix}
    \leftarrow
    \begin{pmatrix}
        y \\
        3yz-x \\
        z
    \end{pmatrix}
    \)
\ENDIF
\ENDWHILE
\end{algorithmic}
For this benchmark, our technique %
generates both closed-forms and invariants in (higher) moments of program variables such as
\(\mathbb{E}(x_n^2 + y_n^2 + z_n^2 - 3x_ny_nz_n)=0\) and
\(\mathbb{E}(x_n-z_n) = 2^{-n}\).
\end{example}

\medskip

The next two benchmarks concern a special class of polynomial automorphisms, the \emph{Yagzhev} maps, (sometimes \emph{cubic-homogeneous} maps) introduced (independently) by Yagzhev \cite{jacobianconj1} and  Bass, Connell, and Wright \cite{jacobianconj2} in the context of the Jacobian conjecture.

Recall that Yagzhev maps \(f\colon \C^n \to \C^n\) take the form \(f(x) = x - g(x)\) such that \(\det f'(x) = 1\) for all \(x\), and \(g\colon \C^n \to \C^n\) is a homogeneous polynomial mapping of degree \(3\).

De Bondt exhibited a Yagzhev mapping in 10 dimensions that has no linear invariants (providing a counterexample to the ``linear dependence conjecture'' for the class of maps) \cite{yagzhev10}.
Zampieri demonstrated that De Bondt's example has quadratic and cubic invariants \cite{yagzhev9} (see  \cref{ex:yagzhev9} for such a Yagzhev map in 9 dimensions).
Work by Santos Freire, Gorni, and Zampieri \cite{yagzhev11} exhibited a Yagzhev mapping in 11 dimensions that has neither linear invariants nor quadratic invariants (see \cref{ex:yagzhev11}).


\begin{example}[\texttt{yagzhev9} \cite{yagzhev9}] \label{ex:yagzhev9} \hfill
\begin{algorithmic}
\WHILE{$\star$}
\STATE \(
    \begin{pmatrix}
        x_1 \\
        x_2 \\
        x_3 \\
        x_4 \\
        x_5 \\
        x_6 \\
        x_7 \\ 
        x_8 \\
        x_9 \\
    \end{pmatrix}
    \leftarrow
    \begin{pmatrix}
        x_1 + x_1 x_7 x_9 + x_2 x_9^2 \\
        x_2 - x_1 x_7^2 -x_2 x_7 x_9\\
        x_3 + x_3x_7x_9 + x_4x_9^2\\
        x_4 - x_3x_7^2 - x_4x_7x_9\\
        x_5 + x_5x_7x_9 + x_6x_9^2\\
        x_6 - x_5x_7^2 - x_6x_7x_9\\
        x_7 + (x_1x_4 - x_2x_3)x_9\\
        x_8 + (x_3x_6 - x_4x_5)x_9\\
        x_9 + (x_1x_4 - x_2x_3)x_8 - (x_3x_6 + x_4x_5)x_7  \\
    \end{pmatrix}
    \)
\ENDWHILE
\end{algorithmic}
For this example, our work generates an invariant quadratic homogeneous polynomial in six variables and three symbolic constants, which we can interpret in terms of determinants (as given below):
    \begin{multline*}
       a\begin{vmatrix}
            x_1(n) & x_2(n) \\
            x_3(n) & x_4(n)
        \end{vmatrix}
        +b\begin{vmatrix}
            x_3(n) & x_4(n) \\
            x_5(n) & x_6(n)
        \end{vmatrix}
       +c\begin{vmatrix}
            x_1(n) & x_2(n) \\
            x_5(n) & x_6(n)
        \end{vmatrix} \\
=
       a\begin{vmatrix}
            x_1(0) & x_2(0) \\
            x_3(0) & x_4(0)
        \end{vmatrix}
    +b\begin{vmatrix}
        x_3(0) & x_4(0) \\
        x_5(0) & x_6(0)
    \end{vmatrix}
    +c\begin{vmatrix}
            x_1(0) & x_2(0) \\
            x_5(0) & x_6(0)
        \end{vmatrix}.
    \end{multline*}
Our computation confirms previous work by the authors of \cite{yagzhev9}.
Those authors demonstrated that this example has no linear invariants, but does admit the above quadratic invariant.
\end{example}

\newpage

\begin{example}[\texttt{yagzhev11} \cite{yagzhev11}] \hfill
\label{ex:yagzhev11}
\begin{algorithmic}
\WHILE{$\star$}
\STATE \(
    \begin{pmatrix}
        x_1 \\
        x_2 \\
        x_3 \\
        x_4 \\
        x_5 \\
        x_6 \\
        x_7 \\ 
        x_8 \\
        x_9 \\
        x_{10} \\
        x_{11}
    \end{pmatrix}
    \leftarrow
    \begin{pmatrix}
        x_1 -x_3 x_{10}^2 \\
        x_2 - x_2x_{11}^2\\
        x_3 + x_1 x_{11}^2 - x_2 x_{10}^2\\
        x_4 - x_6 x_{10}^2\\
        x_5 - x_6 x_{11}^2\\
        x_6 + x_4 x_{11}^2 - x_5 x_{10}^2\\
        x_7 - x_9 x_{10}^2\\
        x_8 - x_9 x_{11}^2\\
        x_9 + x_7 x_{11}^2 - x_8 x_{10}^2 \\
        x_{10} + (x_3 x_5 - x_2 x_6)x_7 + (x_1 x_6 - x_3 x_4)x_8 + (x_2 x_4 - x_1 x_5)x_9\\
        x_{11} - x_{10}^3
    \end{pmatrix}
    \)
\ENDWHILE
\end{algorithmic}
For this example, our approach  generates an invariant cubic homogeneous polynomial, which we can interpret as the determinant of a matrix in the program variables:
    \begin{equation*}
        \begin{vmatrix}
            x_1(n) & x_2(n) & x_3(n) \\
            x_4(n) & x_5(n) & x_6(n) \\
            x_7(n) & x_8(n) & x_9(n)
        \end{vmatrix}
        = 
        \begin{vmatrix}
            x_1(0) & x_2(0) & x_3(0) \\
            x_4(0) & x_5(0) & x_6(0) \\
            x_7(0) & x_8(0) & x_9(0)
        \end{vmatrix}.
    \end{equation*}
    For the avoidance of doubt, the above polynomial was previously found by the authors of \cite{yagzhev11}.  
    Indeed, our implementation confirms previous results: there are neither linear nor quadratic invariants for this example.
\end{example}

\medskip

    \begin{example}[\texttt{nagata} \cite{nagata}] \label{ex:nagata} Let us now consider the classical \emph{Nagata automorphism} introduced by Nagata \cite[pg.\ 41]{nagata} (see also \cite{vandenessen2003polynomial}).
        \begin{algorithmic}
\WHILE{$\star$}
\STATE \(
    \begin{pmatrix}
        x \\
        y \\
        z
    \end{pmatrix}
    \leftarrow
    \begin{pmatrix}
        x - 2(xz+y^2)y - (xz+y^2)^2 z \\
        y + (xz+y^2)z\\
        z
    \end{pmatrix}
    \)
\ENDWHILE
\end{algorithmic}
For the Nagata Automorphism,
it is easy to see that the variable \(z\) is effective and so contributes a linear invariant.
When used to search for quadratic closed-forms, our method generates the polynomial
    \begin{equation*}
        c(x(n)z(n) + y(n)^2) = a z(0) + b z(0)^2 - a z(n) - b z(n)^2 + c (x(0)z(0) + y(0)^2)
    \end{equation*}
where \(a\), \(b\), and \(c\) are symbolic constants.
Here we note that \(x(n)z(n)\) and \(y(n)^2\) are defective monomials.
Our computation confirms the invariants and closed-forms established in \cite{nagata}.

\end{example}

%% file: 09experiments.tex
\section{Experiments}
\label{section:experiments}

In this section, we report on our implementation towards fully automating the analysis of unsolvable loops and describe our experimental setting and results.

\subsection{Implementation}
\hl{Algorithm~\ref{alg:nmcv-alg}, our method for synthesising invariants involving defective variables, and Algorithm~\ref{alg:loop-synthesis}, for the synthesis of solvable loops from unsolvable loops, are implemented in the \Tool{} tool\footnote{\url{https://github.com/probing-lab/polar}}.}
We use  \texttt{python3} and the  \texttt{sympy} package~\cite{10.7717/peerj-cs.103} for symbolic manipulations of algebraic expressions. 

\subsection{Benchmark Selection}

While previous works~\cite{Rodriguez04,Rodriguez07,Oliveira16,Aligator18,BKS19,Kincaid18}
consider invariant synthesis, their techniques are only applicable in a restricted setting:
the analysed loops are, for the most part, solvable;
or, for unsolvable loops, the search for polynomial invariants is template-driven or employs heuristics. 
In contrast, the work herein complements and extends the techniques presented for solvable loops in~\cite{Rodriguez04,Rodriguez07,Oliveira16,Aligator18,BKS19,Kincaid18}.
Indeed, our automated approach turns the problem of polynomial invariant synthesis into a decidable problem for
a larger class of unsolvable loops.  

While solvable loops can clearly be analysed by our work, the main 
benefit of our work comes with handling unsolvable loops by translating them into solvable ones. 
For this reason, in our experimentation we are not interested in examples of solvable loops and so only focus on unsolvable loop benchmarks.  
There is therefore no sensible baseline that we can compare against, as state-of-the-art techniques cannot routinely synthesise invariants for unsolvable loops in the generality we present.

In our work we present a set of \hl{23 examples} of unsolvable loops, as listed in Table~\ref{table:experiments}\footnote{each benchmark in Table~\ref{table:experiments} references, in parentheses,  the respective example from our paper}. 
\hl{Common to all 23} benchmarks from Table~\ref{table:experiments} is the exhibition of circular non-linear dependencies within the variable assignments. We display features of our benchmarks in Table~\ref{table:experiments} (for example, column~3 of Table~\ref{table:experiments} counts the number of defective variables for each benchmark).

Three examples from Table~\ref{table:experiments} are challenging benchmarks taken from the invariant generation literature~\cite{chakarov2016deductive,dreossi2016parallelotope,schreuder2019invariants,sankaranarayanan2020reasoning}; 
full automation in analysing these examples was not yet possible. 
These examples are listed as \texttt{non-lin-markov-1}, \texttt{pts}, and \texttt{bees} in Table~\ref{table:experiments}, respectively corresponding to Example~\ref{ex:NLMarkov-1} (and hence Example~\ref{ex:2:revised}), 
Example~\ref{ex:pts}, and Example~\ref{ex:bees} from Section~\ref{section:case-studies}.
\hl{Eight further benchmarks, as described in Examples~\ref{ex:fib1}--\ref{ex:nagata}, are drawn from the theoretical physics and pure mathematics literature (references are given in Section~\ref{section:case-studies}).
}

\hl{The remaining 12 examples} of Table~\ref{table:experiments} are self-constructed benchmarks to highlight the key ingredients of our work in synthesising  invariants associated with unsolvable recurrence operators.

\begin{table}[t]
  \setlength{\arraycolsep}{3pt}
  \setlength{\tabcolsep}{2.4pt}
  \fontsize{8}{8}\selectfont
  \centering
  \def\arraystretch{1.25}
  \begin{tabular}{@{}lrrrrrr@{}}
    \toprule
    \textsc{Benchmark} & \textsc{Var} & \textsc{Def} & \textsc{Term} & \textsc{Deg} & \textsc{Cand-7} & \textsc{Eqn-7}  \\
    \midrule
    
    \texttt{squares (Fig.~\ref{fig:running:squares})} & 3 & 2 & 8 & 2 & 35 & 113 \\
    \texttt{squares+ (Ex.\ref{ex:squares+})} & 4 & 2 & 12 & 2 & 35 & 204 \\
    \texttt{non-lin-markov-1 (Ex.~\ref{ex:NLMarkov-1})} & 2 & 2 & 11 & 2 & 35 & 64 \\
    \texttt{non-lin-markov-2 (Ex.~\ref{ex:non-lin-markov-2})} & 2 & 2 & 11 & 2 & 35 & 64 \\
    \texttt{prob-squares (Ex.~\ref{fig:running:2})} & 3 & 3 & 4 & 3 & 119 & 337 \\
    \texttt{pts (Ex.~\ref{ex:pts})} & 4 & 2 & 6 & 3 & 35 & 57 \\
    \texttt{squares-squared (Ex.~\ref{ex:squares_solved})} & 4 & 4 & 15 & 4 & 329 & - \\
    \texttt{bees (Ex.~\ref{ex:bees})} & 5 & 5 & 21 & 5 & 791 & - \\
    \texttt{deg-5 (Ex.~\ref{ex:deg-d})} & 3 & 2 & 8 & 5 & 35 & 42 \\
    \texttt{deg-6 (Ex.~\ref{ex:deg-d})} & 3 & 2 & 8 & 6 & 35 & 42 \\
    \texttt{deg-7 (Ex.~\ref{ex:deg-d})} & 3 & 2 & 8 & 7 & 35 & 42 \\
    \texttt{deg-8 (Ex.~\ref{ex:deg-d})} & 3 & 2 & 8 & 8 & 35 & 43 \\
    \texttt{deg-9 (Ex.~\ref{ex:deg-d})} & 3 & 2 & 8 & 9 & 35 & 43 \\
    \texttt{deg-500 (Ex.~\ref{ex:deg-d})} & 3 & 2 & 8 & 500 & 35 & 43 \\
    \hl{\texttt{fib1 (Ex.~\ref{ex:fib1})}} & 3 & 3 & 4 & 2 & 119 & 204  \\
    \hl{\texttt{fib2 (Ex.~\ref{ex:fib2})}} & 3 & 3 & 7 & 3 & 119 & 368 \\
    \hl{\texttt{fib3 (Ex.~\ref{ex:fib3})}} & 3 & 3 & 7 & 2 & 119 & 204 \\
    \hl{\texttt{markov-triples-toggle (Ex.~\ref{ex:markov-triples-toggle})}} & 3 & 3 & 10 & 2 & 119 & 454 \\
    \hl{\texttt{markov-triples-random (Ex.~\ref{ex:markov-triples-random})}} & 3 & 3 & 9 & 2 & 119 & 254 \\
    \hl{\texttt{yagzhev9 (Ex.~\ref{ex:yagzhev9})}} & 9 & 9 & 29 & 3 & 11439 & \timeout \\
    \hl{\texttt{yagzhev11 (Ex.~\ref{ex:yagzhev11})}} & 11 & 11 & 30 & 3 & 31823 & \timeout  \\
    \hl{\texttt{nagata (Ex.~\ref{ex:nagata})}} & 3 & 2 & 9 & 4 & 35 & 1313 \\
    \bottomrule
  \end{tabular}
  \vspace{0.5em}
  \caption{Features of the benchmarks. \textsc{Var} = Total number of loop variables; \textsc{Def} = Number of defective variables; \textsc{Term} = Total number of terms in assignments; \textsc{Deg} = Maximum degree in assignments; \textsc{Cand-7} = Number of monomials in candidate with degree $7$; \textsc{Eqn-7} = Size of the system of equations associated with a candidate of degree $7$; \textsc{-} = Timeout (60 seconds).}
\label{table:experiments}
\end{table}
\subsection*{Experimental Setup} 
We evaluate our approach in \Tool{} on the examples from  Table~\ref{table:experiments}. All our  experiments were performed on a machine with a 1.80GHz Intel i7 processor and 16 GB of RAM.

\subsection*{Evaluation Setting}
The landscape of benchmarks in the invariant synthesis literature for solvable loops can appear complex with: high numbers of variables, high degrees in polynomial updates, and multiple update options.
However, we do not intend to compete on these metrics for solvable loops. 
The power of our invariant synthesis technique lies in its ability to handle `unsolvable' loop programs: those with cyclic inter-dependencies and non-linear self-dependencies in the loop body.
\hl{
Regarding Algorithm~\ref{alg:loop-synthesis}, specifying our method for synthesising solvable from unsolvable loops, the main complexity lies in constructing an invariant involving defective variables.
Once, such an invariant has been found, the remaining complexity of Algorithm~\ref{alg:loop-synthesis} is linear in the number of program variables.
Hence, the experimental results for our invariant synthesis technique also show the feasibility of our new method synthesising solvable from unsolvable loops.
}
While  the benchmarks of Table~\ref{table:experiments} may  be considered {simple}, the fact that previous works cannot systematically handle such \emph{simple models} crystallises that even simple loops can be unsolvable, limiting the applicability of state-of-the-art methods, as illustrated in the example below. 

\medskip

\begin{example}
    Consider the question: \emph{does the unsolvable loop program \texttt{deg-9} in Table~\ref{table:experiments} (i.e. Example~\ref{ex:deg-d}) possess a cubic invariant?}
    The program variables for \texttt{deg-9} are \(x,y,\) and \(z\).
    The variables \(x\) and \(y\) are  defective. 
    Using \Tool{}, we derive that 
   the cubic, non-trivial polynomial  \(p(x_n,y_n,z_n)\) given by 
        \begin{multline*}
             12(ay_n + by_n^2 + cy_n^3 + dx_n +ex_ny_n + fx_ny_n^2) - (3a+24b+117c+2d+17e+26f)x_n^2
             \\ - (6a-6b+315c+4d-2e+88f)x_n^2y_n + 3(3a-3b+144c+2d-e+35f)x_n^3
        \end{multline*}
        yields a cubic polynomial loop invariant, where \(a,b,c,d,e,\) and \(f\) are symbolic constants. 
    Moreover, for \(n\ge 1\), the expectation of this polynomial  (\texttt{deg-9} is a probabilistic loop) 
    in the \(n\)th iteration is given by
    \begin{equation*}
        \E(p(x_n,y_n,z_n)) = -108a +312b -1962c -68d +52e -68f.
    \end{equation*}
\end{example}

\begin{table}[t]
  \setlength{\arraycolsep}{3pt}
  \setlength{\tabcolsep}{2.4pt}
  \fontsize{8}{8}\selectfont
  \centering
  \def\arraystretch{1.25}
  \begin{tabular}{@{}llrrrrrrrrrrrrrrr@{}}
    \toprule
    \multicolumn{1}{c}{} & \phantom{x} & \multicolumn{7}{c}{Candidate Degree}\\ 
    \cmidrule{1-2} \cmidrule{3-9}
    \textsc{Benchmark} && 1 & 2 & 3 & 4 & 5 & 6 & 7  \\
    \midrule
    
    \texttt{squares (Fig.~\ref{fig:running:squares})} && *0.95 & 1.09 & 1.66 & 2.72 & 4.98 & 11.67 & 19.87 & \\
    \texttt{squares+ (Ex.~\ref{ex:squares+})} && *0.67 & 1.02 & 1.92 & 3.77 & 7.39 & 15.66 & 28.78  \\
    \texttt{non-lin-markov-1 (Ex.~\ref{ex:NLMarkov-1})} && *0.42 & *0.68 & *1.04 & *2.35 & *4.02 & *9.81 & *13.58  \\
    \texttt{non-lin-markov-2 (Ex.~\ref{ex:non-lin-markov-2})} && *0.38 & *0.59 & *1.15 & *2.40 & *3.19 & *5.91 & *14.91  \\
    \texttt{prob-squares (Ex.~\ref{ex:prob-squares})} && *0.76 & 1.50 & 4.69 & 20.85 & \timeout & \timeout & \timeout  \\
    \texttt{squares-and-cube (Fig.~\ref{fig:running:2})} && 0.32 & *0.51 & *1.20 & *4.21 & *19.49 & \timeout & \timeout \\
    \texttt{pts (Ex.~\ref{ex:pts})} && *0.35 & *0.49 & *0.71 & *1.13 & *1.90 & *3.42 & *5.98 \\
    \texttt{squares-squared (Ex.~\ref{ex:squares-squared})} && *0.48 & *1.55 & *8.92 & \timeout & \timeout & \timeout & \timeout \\
    \texttt{bees (Ex.~\ref{ex:bees})} && *0.69 & *3.27 & *42.16 & \timeout & \timeout & \timeout & \timeout  \\
    \texttt{deg-5 (Ex.~\ref{ex:deg-d})} && *0.46 & *0.65 & *1.01 & *1.94 & *4.04 & *8.20 & *19.28 \\
    \texttt{deg-6 (Ex.~\ref{ex:deg-d})} && *0.54 & *0.69 & *1.62 & *1.91 & *4.32 & *8.24 & *19.75 \\
    \texttt{deg-7 (Ex.~\ref{ex:deg-d})} && *0.49 & *0.98 & *1.06 & *1.84 & *4.42 & *8.98 & *19.39 \\
    \texttt{deg-8 (Ex.~\ref{ex:deg-d})} && *0.45 & *0.62 & *1.08 & *2.04 & *4.07 & *8.93 & *20.97 \\
    \texttt{deg-9 (Ex.~\ref{ex:deg-d})} && *0.47 & *0.65 & *1.11 & *1.84 & *4.31 & *7.85 & *19.67 \\
    \texttt{deg-500 (Ex.~\ref{ex:deg-d})} && *0.47 & *0.65 & *1.08 & *1.96 & *4.29 & *8.58 & *21.05 \\
    \hl{\texttt{fib1 (Ex.~\ref{ex:fib1})}} && 0.31 & 0.41 & *0.91 & 2.02 & 2.97 & *5.47 & 14.66 \\
    \hl{\texttt{fib2 (Ex.~\ref{ex:fib2})}} && 0.31 & 0.49 & *1.34 & 3.06 & 8.64 & *17.24 & \timeout \\
    \hl{\texttt{fib3 (Ex.~\ref{ex:fib3})}} && 0.31 & 0.48 & *1.73 & 3.46 & 16.92 & \timeout & \timeout \\
    \hl{\texttt{markov-triples-toggle (Ex.~\ref{ex:markov-triples-toggle})}} && 0.39 & 0.61 & *1.41 & 2.63 & 5.84 & *9.45 & 22.71 \\
    \hl{\texttt{markov-triples-random (Ex.~\ref{ex:markov-triples-random})}} && *0.44 & *0.62 & *1.65 & *3.52 & *7.93 & *21.39 & *71.23 \\
    \hl{\texttt{yagzhev9 (Ex.~\ref{ex:yagzhev9})}} && 0.47 & *2.44 & *24.31 & \timeout & \timeout & \timeout & \timeout \\
    \hl{\texttt{yagzhev11 (Ex.~\ref{ex:yagzhev11})}} && 0.59 & 7.12 & *39.27 & \timeout & \timeout & \timeout & \timeout \\
    \hl{\texttt{nagata (Ex.~\ref{ex:nagata})}} && 0.33 & *1.04 & *3.49 & *8.19 & *40.03 & \timeout & \timeout \\
    \bottomrule
  \end{tabular}
  
  \vspace{0.5em}
  \timeout{} = Timeout (75 seconds); {}* = Found invariant of the corresponding degree.
  \vspace{0.5em}
  \caption{\hl{The time taken to search for polynomial candidates with closed-forms (results in seconds)}. 
}
     \label{table:experiments2}
\end{table}

\subsection*{Experimental Results} 
Our experiments using \Tool{} to synthesise invariants are summarised in  Table~\ref{table:experiments2}, using the 
examples of Table~\ref{table:experiments}. 
Patterns in Table~\ref{table:experiments2} show that, if time considerations are the limiting factor, then the greatest impact cannot be attributed to the number of program variables nor the maximum degree in the program assignments (Table~\ref{table:experiments}).
Three of the examples in Table~\ref{table:experiments} exhibit timeouts (60 secs) in the final column.
The property common to each of these examples is the high number of monomial terms in any polynomial candidate of degree 7. In turn, this property feeds into a large system of simultaneous equations, which we solve to test for invariants.
Indeed, time elapsed is not so strongly correlated with either of these program features.
As supporting evidence we note the specific attributes of benchmark \texttt{deg-500} whose assignments include polynomial updates of large degree and yet returns synthesised invariants with relatively low time elapsed in Table~\ref{table:experiments2}.
We note the significantly longer running times associated with the benchmark \texttt{bees} (Example~\ref{ex:bees}).
This suggests that mutual dependencies between program variables in the loop assignment  explain this phenomenon: such inter-relations lead to the construction of larger systems of equations, which itself feeds into the problem of resolving the recurrence equation associated with a candidate.

\subsection*{Experimental Summary}
Our experiments illustrate the feasibility of synthesising invariants \hl{and solvable loops} using our approach for programs with unsolvable recurrence operators from various domains such as biological systems, probabilistic loops, and classical programs (see Section~\ref{sec:invariants}).
This further motivates the theoretical characterisation of unsolvable operators in terms of defective variables (Section~\ref{sec:effective-deffective}).

%% file: 10conclusion.tex
\section{Conclusion}
\label{sec:conclusion}

We establish a new technique that synthesises invariants for loops with unsolvable recurrence operators and show its applicability for deterministic and probabilistic programs. \hl{Our work is further extended to translate unsolvable loops into solvable ones, by ensuring that the polynomial loop invariants of the solvable loop are invariants of the given unsolvable loop.} 
Our work is based on a new characterisation of unsolvable loops in terms of effective and defective variables: the presence of defective variables is equivalent to unsolvability.
In order to generate invariants, we provide an algorithm to isolate the defective program variables and a new method to compute polynomial combinations of defective variables admitting exponential polynomial closed-forms.
The implementation of our approach in the tool \Tool{} and our experimental evaluation demonstrate the usefulness of our alternative characterisation of unsolvable loops and the applicability of our invariant synthesis technique to systems from various domains.